\documentclass[11pt]{article}

\usepackage[T1]{fontenc}
\usepackage[utf8]{inputenc}
\usepackage{lmodern}
\usepackage[a4paper,margin=1in]{geometry}
\usepackage{microtype}
\usepackage{amsmath,amssymb,amsthm,mathtools}
\usepackage{booktabs}
\usepackage{enumitem}
\usepackage{xcolor}
\usepackage{graphicx}
\usepackage{float}
\usepackage{tikz}
\usetikzlibrary{arrows.meta,positioning}
\usepackage{multirow}
\usepackage{todonotes}
\usepackage{aliascnt}
\usepackage{hyperref}
\usepackage[nameinlink,capitalise,noabbrev]{cleveref}

\hypersetup{
  colorlinks=true,
  linkcolor=blue!55!black,
  citecolor=green!40!black,
  urlcolor=blue!65!black,
  pdftitle={A Near-Linear Element-Kernel for d-Hitting Set}
}

\theoremstyle{plain}
\newtheorem{theorem}{Theorem}[section]

\newaliascnt{lemma}{theorem}
\newtheorem{lemma}[lemma]{Lemma}
\aliascntresetthe{lemma}

\newaliascnt{corollary}{theorem}

\aliascntresetthe{corollary}

\newaliascnt{proposition}{theorem}
\newtheorem{proposition}[proposition]{Proposition}
\aliascntresetthe{proposition}

\theoremstyle{definition}

\newaliascnt{definition}{theorem}
\newtheorem{definition}[definition]{Definition}
\aliascntresetthe{definition}

\newaliascnt{remark}{theorem}

\aliascntresetthe{remark}

\crefname{theorem}{Theorem}{Theorems}
\Crefname{theorem}{Theorem}{Theorems}
\crefname{lemma}{Lemma}{Lemmas}
\Crefname{lemma}{Lemma}{Lemmas}
\crefname{corollary}{Corollary}{Corollaries}
\Crefname{corollary}{Corollary}{Corollaries}
\crefname{proposition}{Proposition}{Propositions}
\Crefname{proposition}{Proposition}{Propositions}
\crefname{definition}{Definition}{Definitions}
\Crefname{definition}{Definition}{Definitions}
\crefname{remark}{Remark}{Remarks}
\Crefname{remark}{Remark}{Remarks}

\newcommand{\HS}{\ensuremath{\text{\normalfont\scshape \(d\)-Hitting Set}}}

\title{A Near-Linear Element-Kernel for \(d\)-Hitting Set}

\author{
Zimo Sheng$^{1}$, Mingyu Xiao$^{2}$\\[0.5em]
$^{1}$Anhui University\\
$^{2}$University of Electronic Science and Technology of China\\[0.5em]
\href{mailto:shengzimo2016@gmail.com}
{\texttt{shengzimo2016@gmail.com}}
\qquad
\href{mailto:myxiao@gmail.com}
{\texttt{myxiao@gmail.com}}\\[0.25em]
}
\date{}
\begin{document}
\maketitle
\begin{abstract}
In \(d\)-\textsc{Hitting Set}, the input consists of a finite universe \(U\), a family \(\mathcal S\) of subsets of \(U\) with size at most \(d\), and an integer \(k\). The task is to decide whether at most \(k\) elements of \(U\) can intersect every set in \(\mathcal S\).  Previously known kernels contain \(O(k^{d-1})\) elements and \(O(k^d)\) sets. Whether the number of elements can be reduced to \(O(k^{d-1-\varepsilon})\) for some constant
\(\varepsilon>0\)
 is a longstanding open problem in kernelization. 
In this paper, we answer this question affirmatively by giving  a  near-linear   randomized element-kernel and a subquadratic deterministic element-kernel  through a re-encoding of the instance.
For every fixed \(d\geq3\), we give a one-sided randomized kernel with   \(O_d(k\log^3k)\)  elements and a deterministic kernel with   \(O_d(k^{7/4}\log^{3/2}k)\)  elements.
Our kernels may still contain \(k^{O(d)}\) sets and the parameter \(k\)  may grow polynomially.
 \end{abstract}
\section{Introduction}
An instance $I=(U, \mathcal S, k)$ of \HS{} consists of a finite
universe \(U\), a family \(\mathcal S\subseteq 2^U\) in which every set
has size at most \(d\), and a nonnegative integer \(k\). The question is
whether there is a set \(Z\subseteq U\) of size at most \(k\) that
intersects every member of \(\mathcal S\).  This is a central
covering problem in parameterized complexity and it provides a common
formulation for many finite-obstruction deletion problems~\cite{Cai1996}.
In this paper, we will assume that \(d\) is a fixed constant and study
kernelization for \(d\)-\textsc{Hitting Set}.
The number of elements and the number of sets in \HS{}
are two quantities frequently used to measure kernel size. The earliest known
\(O(k^d)\)-set kernel is due to Fellows et al.~\cite{fellows2008faster,FominEtAl2023}. The now-standard sunflower-based formulation can be traced to Flum and Grohe~\cite{DBLP:series/txtcs/FlumG06} and appears in later textbooks
\cite{CyganEtAl2015,FominEtAl2019Book}.
Reduction
rules due to Abu-Khzam, together with their later combination with
sunflower reductions by van Bevern, give kernels with
\(O(k^{d-1})\) elements and \(O(k^d)\) sets
\cite{AbuKhzam2010,vanBevern2014}. The same element bound was also
obtained by Moser~\cite{Moser2010}. Liu and Xiao recently improved the
leading coefficient in the element bound without changing its
exponent \cite{LiuXiao2025}.
These two bounds address different questions. Dell and van Melkebeek
exclude compressions of bit size \(O(k^{d-\varepsilon})\) for every
fixed \(d\) and \(\varepsilon>0\), unless
\(\mathsf{coNP}\subseteq\mathsf{NP/poly}\)
\cite{DellVanMelkebeek2014}. This sparsification barrier makes the
\(O(k^d)\) information bound essentially tight, but it does not rule
out an equivalent instance with far fewer distinct elements and
polynomially many bounded-size sets. Whether the \(O(k^{d-1})\) element bound can be improved by a
polynomial factor has remained a longstanding open question
\cite{worker2010,fptSchool2014,FominEtAl2023}. A central formulation asks whether
the degree of the element bound can be made independent of the fixed
  set-size bound  \(d\) \cite{FominEtAl2019Book}. Stronger questions
ask for a linear element bound
\cite{fptSchool2014}.
 Jansen's recent survey \cite{Jansen2026QuoVadis} lists this as
Problem~5 and notes that even for \(d=3\), no subquadratic
element bound was known.
The   known \(O(k^{d-1})\) element bound was  mainly obtained through
combinatorial reductions, including crown techniques
\cite{AbuKhzam2010,vanBevern2014} and sunflower techniques
\cite{ErdosRado1960}.
Some of these reductions also replace input sets by smaller subsets
\cite{AbuKhzam2010,LiuXiao2025}.
If we allow only the deletion of elements and whole sets, the output
is a subhypergraph of the input and we may call such kernels \textit{subhypergraph kernels}. For every fixed \(d\geq3\) and
\(0<\varepsilon<d-1\), a deterministic subhypergraph kernel with
\(O(k^{d-1-\varepsilon})\) elements would imply
\(\mathsf{coNP}\subseteq\mathsf{NP/poly}\), as shown in
\Cref{prop:induced-lower-bound}.
This limitation motivates us to ask whether   changing the representation
of the instance could yield a kernel with fewer elements.
Note that Bannach and Tantau used color-coding~\cite{AlonYusterZwick1995} to compute
Hitting Set kernels by constant-depth circuits \cite{BannachTantau2018}.
This also suggests that hashing might be 
a good way to encode and reduce the number of elements in a kernel.
 Our idea is to use fewer Boolean
variables to represent the elements of a solution by hashing, and use CNF clauses
to ensure that this representation describes a valid hitting set. With the help of this idea, finally, we obtain our main results.
For every fixed \(d\geq3\), we give a one-sided randomized kernel with
  \(O_d(k\log^3k)\) elements and \(O_d(k^d\log^{2d-2}k)\)   sets, as well as
a deterministic kernel with     \(O_d(k^{7/4}\log^{3/2}k)\) elements and
\(O_d(k^{d+3/2}\log k)\)  sets.
Our randomized kernel has a near-linear number of elements, while our
deterministic kernel breaks the quadratic element bound for
\(3\)-Hitting Set.
 Our compression still allows at most \(k\) input elements to be selected. The parameter grows when we reduce it back to \HS{}.
Representative graph deletion problems that can be formulated as
\(d\)-Hitting Set include
Cluster Vertex Deletion~\cite{HueffnerEtAl2010,BoralEtAl2016},
Feedback Vertex Set in Tournaments~\cite{KumarLokshtanov2016},
Subset Feedback Vertex Set in chordal and split graphs~\cite{PhilipEtAl2019},
Split Vertex Deletion~\cite{GhoshEtAl2015Split,AgrawalEtAl2020Split}, and
\(d\)-Path Vertex Cover~\cite{CervenyChoudharySuchy2022}.
Kernelization has been extensively studied for special cases of
\(d\)-Hitting Set \cite{AhnKimKwon2026,BodlaenderEtAl2016Meta,FominEtAl2019Implicit}
and related \(d\)-Set Packing problems
\cite{bessy2023kernelization,DellMarx2012,JansenRoy2023}.
Another related development is lossy kernelization
\cite{lokshtanov2017lossy}, which achieves  \(O(k)\)  elements for fixed
\(d\)  by allowing a controlled approximation loss   \cite{FominEtAl2023}.
 Throughout the paper, \(d\geq3\)
is fixed, and the constants hidden in \(O_d(\cdot)\) may depend on \(d\).
\subsection{Results}
\label{sec:intro-results}
Our main results are two kernels for \HS{}, one randomized and one deterministic.
The randomized kernel has one-sided error. It maps every NO-instance
to a NO-instance with probability one, while every fixed YES-instance
is mapped to a YES-instance with probability at least \(3/4\).
They are presented in the following two theorems. Their element and set
bounds and output parameters are summarized in Table~\ref{tab:results}.
\begin{theorem}
\label{thm:intro-random-two-level}
For every fixed \(d\geq3\), \HS{} admits a one-sided randomized kernel   computable in expected polynomial time with  success probability at least \(3/4\).  The output has \(O_d(k\log^3k)\) elements, \(O_d(k^d\log^{2d-2}k)\) sets, and parameter \(k'=O_d(k\log^3k)\).
 \end{theorem}
\begin{theorem}
\label{thm:intro-deterministic}
For every fixed \(d\geq3\), \HS{} admits a deterministic kernel with
    \(O_d(k^{7/4}\log^{3/2}k)\) elements, \(O_d(k^{d+3/2}\log k)\)   sets,
and output parameter     \(k'=O_d(k^{7/4}\log^{3/2}k)\) .
\end{theorem}
\begin{table}[H]
\centering
\renewcommand{\arraystretch}{1.4}
\setlength{\tabcolsep}{5pt}
    \resizebox{\linewidth}{!}{%
\begin{tabular}{@{}ccccc@{}}
\toprule
Results & Kernel types & \# of elements & \# of sets
& Output parameter \(k'\)\\
\midrule
\Cref{thm:intro-random-two-level}  (proofs in Sec.~\ref{sec:framework} 
\&~\ref{sec:randomized-kernels})
&
one-sided randomized
&
\(\displaystyle O_d\!\left(k\log^3k\right)\)
&
\(\displaystyle O_d\!\left(k^d\log^{2d-2}k\right)\)
&
\(\displaystyle O_d\!\left(k\log^3k\right)\)
\\
\addlinespace[3pt]
\Cref{thm:intro-deterministic} (proofs in Sec.~\ref{sec:framework} 
\&~\ref{sec:det})
&
deterministic
&
\(\displaystyle O_d\!\left(k^{7/4}\log^{3/2}k\right)\)
&
\(\displaystyle O_d\!\left(k^{d+3/2}\log k\right)\)
&
\(\displaystyle O_d\!\left(k^{7/4}\log^{3/2}k\right)\)
\\
\bottomrule
\end{tabular}%
}
  \caption{Element and set bounds and output parameters of the two kernels.}
\label{tab:results}
\end{table}
\section{Framework of the Algorithm}
\label{sec:framework}
\label{sec:technical-overview}
Our goal is to reduce the number of elements in a kernel. To this end,
we first encode the input as a CNF formula with few variables and then
reduce the formula back to \HS{}. In the final reduction, each Boolean
variable produces two elements, and each clause produces one set.
Thus, reducing the number of variables in the CNF formula reduces the number of output
elements, while limiting each clause to at most \(d\) literals ensures
that the output is still an instance of \HS{}.
Our kernelization algorithm consists of three major steps as shown in
Figure~\ref{fig:kernelization-pipeline}. The first step  applies a known kernelization algorithm to reduce the input to an
equivalent instance \(\mathcal I_1\) with \(O(k^{d-1})\) elements and
\(O(k^d)\) sets. The second step will encode the instance \(\mathcal I_1\) into a $d$-CNF formula $\mathcal I_2=\Phi$. It first uses a packing to divide the elements into a part represented directly and a part encoded by hashing, and then
uses either
randomized or deterministic hashing to represent a hitting set with few
Boolean variables, with clauses of width at most \(d\). The third step reduces the resulting formula back to \HS{}.
\begin{figure}[ht]
\centering
\begin{tikzpicture}[
  instance/.style={
    draw=black!65,
    line width=0.55pt,
    rounded corners=2.5pt,
    minimum width=3.0cm,
    minimum height=2.35cm,
    inner sep=0pt
  },
  hittinginstance/.style={
    instance,
    fill=green!9
  },
  cnfinstance/.style={
    instance,
    draw=blue!55!black,
    line width=0.9pt,
    fill=blue!7
  },
  probleminfo/.style={
    anchor=north,
    align=center,
    text width=2.85cm,
    inner sep=3pt,
    font=\small
  },
  sizeinfo/.style={
    anchor=south,
    align=center,
    text width=2.85cm,
    inner sep=2pt,
    font=\small
  },
  flowarrow/.style={
    -{Latex[length=2mm]},
    line width=0.75pt,
    draw=black!70
  }
]
  \node[hittinginstance] (i0) {};
  \node[probleminfo] at (i0.north) {
    \(d\)-\textsc{Hitting Set}\\[2pt]
    {\footnotesize
    \(\boldsymbol{\mathcal I_0=(U,\mathcal S,k)}\)}
  };
  \node[sizeinfo] at (i0.south) {
    {\footnotesize \(|U|\) elements}\\
    {\footnotesize \(|\mathcal S|\) sets}
  };
  \node[hittinginstance,right=10mm of i0] (i1) {};
  \node[probleminfo] at (i1.north) {
    \(d\)-\textsc{Hitting Set}\\[2pt]
    {\footnotesize
    \(\boldsymbol{\mathcal I_1=(U_1,\mathcal S_1,k_1)}\)}
  };
   \node[sizeinfo] at (i1.south) {
    {\footnotesize \(|U_1|=O_d(k^{d-1})\)}\\
    {\footnotesize \(|\mathcal S_1|=O_d(k^d)\)}\\
    {\footnotesize \(k_1\leq k\)}
  };
  \node[cnfinstance,right=10mm of i1] (i2) {};
  \node[probleminfo] at (i2.north) {
    \(d\)-CNF\\[2pt]
    {\footnotesize
    \(\boldsymbol{\mathcal I_2=\Phi}\)}
  };
  \node[sizeinfo] at (i2.south) {
    {\footnotesize \(N\) variables}\\
    {\footnotesize \(M\) clauses}
  };
  \node[hittinginstance,right=10mm of i2] (i3) {};
  \node[probleminfo] at (i3.north) {
    \(d\)-\textsc{Hitting Set}\\[2pt]
    {\footnotesize
    \(\boldsymbol{\mathcal I_3=(U',\mathcal S',k')}\)}
  };
    \node[sizeinfo] at (i3.south) {
    {\footnotesize \(|U'|=2N\)}\\
    {\footnotesize \(|\mathcal S'|\leq N+M\)}\\
    {\footnotesize \(k'=N\)}
  };
  \draw[flowarrow]
    (i0) -- node[above,font=\scriptsize]{Step~1} (i1);
  \draw[flowarrow]
    (i1) -- node[above,font=\scriptsize]{Step~2} (i2);
  \draw[flowarrow]
    (i2) -- node[above,font=\scriptsize]{Step~3} (i3);
\end{tikzpicture}
\caption{The three-step kernelization algorithm starting from the
original input instance \(\mathcal I_0\).}
\label{fig:kernelization-pipeline}
\end{figure}
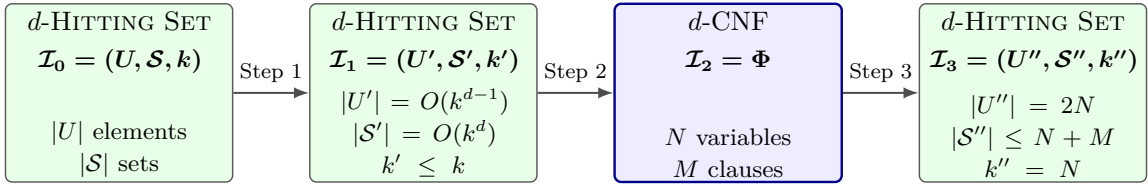
The main contribution of our algorithm is the two encodings introduced in the second step. The proofs of the corresponding lemmas are deferred to \Cref{sec:randomized-kernels} and \Cref{sec:det}, respectively.
 In the next three subsections, we introduce the three steps of the algorithm, respectively.  We then use these encodings to derive our two main theorems.
\subsection{Step 1: Reducing the Input Instance}
\label{sec:framework-preprocessing}
We first bound the input size so that both the length of element
identifiers and the number of sets to be encoded in Step 2 can be bounded in
terms of \(k\). We only need to use the following known kernel for
\HS{} in the literature~\cite{AbuKhzam2010,vanBevern2014,LiuXiao2025}.
\begin{lemma}\label{lem:reduced-instance}
For every fixed \(d\), \HS{} admits a kernel
with \(O_d(k^{d-1})\) elements and \(O_d(k^d)\) sets, without increasing the parameter $k$.
\end{lemma}
We first apply   a kernelization that satisfies \mbox{
\Cref{lem:reduced-instance} }\hskip0pt
and does not increase \(k\)~\mbox{
\cite{vanBevern2014}}\hskip0pt
 .
We continue to denote the reduced instance by \((U,\mathcal S,k)\), and write
\[
n=|U|=O_d(k^{d-1}),\qquad m=|\mathcal S|=O_d(k^d).
\]

 We also handle the trivial cases as follows.
If \(\emptyset\in\mathcal S\), the instance is a NO-instance.
If \(\mathcal S=\emptyset\), it is a YES-instance.
If \(k=0\) and \(\mathcal S\neq\emptyset\), it is a NO-instance.
If \(\emptyset\notin\mathcal S\) and \(k\geq |U|\), the set
\(U\) is a feasible solution and thus the instance is a YES-instance.
For \(k=1\), we check whether the input sets have a common element.
In these cases, we return the corresponding constant-size instance and stop.
Henceforth, we may assume that
\[
\emptyset\notin\mathcal S,\qquad \mathcal S\neq\emptyset,\qquad
2\leq k<|U|.
\]

\subsection{Step 2: Constructing the CNF Formula}
\label{sec:framework-encoding}
 \label{sec:packing}

  In this step, we encode the  instance \(\mathcal I_1\)   obtained in Step 1 as a \(d\) -CNF formula   \(\Phi\) . Note that for the instance \(\mathcal I_1\), each element identifier now needs only \(O_d(\log k)\) bits.

Both our randomized and deterministic encodings next use the same packing preprocessing to find a set \(X\) of at most \(dk\) elements that intersects every input set.
A packing is a subfamily of pairwise disjoint input sets. We greedily
construct an inclusion-maximal packing
\(\mathcal P\subseteq\mathcal S\). If \(|\mathcal P|>k\), we return a
fixed   unsatisfiable CNF formula and stop the encoding , since a hitting set must use a different element to
hit each member of \(\mathcal P\). Thus, we may assume that
\(|\mathcal P|\leq k\).

Define
\[
X=\bigcup_{A\in\mathcal P}A,
\qquad
W=U\setminus X.
\]
Since every set has size at most \(d\),
\[
|X|\leq d|\mathcal P|\leq dk.
\]

Maximality of \(\mathcal P\) implies that every input set intersects
\(X\). Otherwise, that set could be added to \(\mathcal P\).
Consequently,
\begin{equation}\label{eq:outside-X}
|A\cap W|\leq d-1
\qquad
\text{for every }A\in\mathcal S.
\end{equation}

  We then  represent the elements of \(X\) directly by variables and use hashing
to encode the selected elements of \(W\). Equation~\eqref{eq:outside-X}
ensures that at most \(d-1\) elements of each input set need this
encoding, which is crucial for keeping the clause width at most \(d\).

The elements of \(X\) need only \(O_d(k)\) variables for a direct
representation, but \(W\) may still contain \(O_d(k^{d-1})\) elements.
We therefore record only the elements of \(W\) selected by a solution,
allowing several candidate elements to share the same variables.
The CNF constraints must ensure that at most \(k\) elements are
represented in total and that they hit every input set.

We give randomized and deterministic hash-based encodings for this task.
Let \(N\) and \(M\) denote the total numbers of variables and clauses
in the resulting formula.

\paragraph{Randomized encoding.}
We have the following main lemma with the proof deferred to \Cref{sec:randomized-kernels}.

\begin{lemma}[Randomized encoding]\label{lem:random-cnf}
  Given a  \HS{} instance \((U,\mathcal S,k)\) preprocessed as in
\Cref{sec:framework-preprocessing},   a randomized algorithm constructs
a \(d\)-CNF formula \(\Phi\)  in expected polynomial time  .
The formula has  \(N\) variables and \(M\) clauses,  where
\[
N=O_d(k\log^3 k),\qquad M=O_d(k^d\log^{2d-2}k).
\]
If the input is a NO-instance, then \(\Phi\) is unsatisfiable .
If the input is a YES-instance, then \(\Phi\) is satisfiable
with probability at least \(3/4\).
\end{lemma}

We use two levels of hashing to   record the elements selected from
 \(W\)   in  a shared table. The  algorithm samples one first-level function
  to partition \(W\) into buckets, and then  a small number of candidate
second-level functions  to divide each bucket into subbuckets . The CNF
formula then chooses a second-level function  for each first-level
bucket   to place its selected elements in distinct subbuckets.   We arrange the subbuckets in a table, with one row per first-level bucket
and one column per subbucket number. At the position corresponding to a
subbucket, Boolean variables store the  selected element's   binary
identifier, which we call its address .

We only need to separate elements selected together in one solution.
Thus, several candidate elements may   fall into the same subbucket, and
different assignments to its variables can record different elements .
The CNF formula   requires that the elements selected directly from \(X\)
and those represented by the table from \(W\) total at most \(k\), and
 checks that every input set contains a   selected  element.

Specifically, the first level uses about \(k/\log k\) buckets, leaving
\(O(\log k)\) solution elements in each bucket with high probability.
The second level therefore only needs \(O(\log^2 k)\) subbuckets per
bucket and \(O(\log k)\) randomly sampled candidate functions.   There are  \(O(k\log k)\)   subbuckets in total , each using \(O_d(\log k)\)
variables to record an element. The   checks on these variables  are encoded
separately for each candidate second-level function,   using
\(O_d(k\log^3 k)\) variables in total. The  auxiliary variables for
function selection and   for bounding the number of selected elements
also fit within this bound. Thus, the entire formula uses
\(O_d(k\log^3 k)\) variables, as required by  \Cref{lem:random-cnf}.

\paragraph{Deterministic encoding.} 
We have the following main lemma with the proof deferred to \Cref{sec:det}.

\begin{lemma}[Deterministic encoding]\label{lem:deterministic-cnf}
For every \HS{} instance \((U,\mathcal S,k)\) preprocessed as in
\Cref{sec:framework-preprocessing}, one can construct in deterministic
polynomial time a CNF formula \(\Phi\) with \(N\) variables and \(M\)
clauses, each of width at most \(d\), where
\[
N=O_d(k^{7/4}\log^{3/2}k),\qquad M=O_d(k^{d+3/2}\log k).
\]
The input is a YES-instance if and only if \(\Phi\) is satisfiable.
\end{lemma}

  We use a single level of hashing to partition \(W\) into buckets and
  provide several storage positions in each bucket. Each position has an
activation variable and identifier variables. Their assignment determines
whether the position is used and which selected element it stores.
Several selected elements in the same bucket can thus occupy different
positions . The algorithm constructs an explicit   family of candidate
functions  such that, for   every  set of at most \(k\) elements of \(W\),
some function   lets all these elements fit in their buckets . The CNF
formula chooses one function   and requires the number of active positions
plus the number of  directly selected elements of \(X\)   to be at most  \(k\).

  After storing these identifiers, we need to check whether each input set
contains a selected element. For sets not already hit by directly selected
elements  of \(X\), we   want candidate elements to share the variables
used for this check rather than introducing a match variable for each
element. We therefore  represent element identifiers as low-degree
polynomials and   substitute the same set of distinct values, called
evaluation points, into them. Each bucket shares a table recording the
polynomial  values of its   active records at these points . At each point,
we check whether some element of the input set has   its value in the
table of its bucket, expressing the hitting condition with clauses of
width at most \(d\). A stored  element of the input set   provides a match
at every point. Otherwise, matches arise only where distinct polynomials
happen to agree. Since the total number of such collisions is bounded,
 sufficiently many evaluation points   ensure that every unhit set fails
 at least one check .

Specifically, with \(B\) buckets, the candidate family contains \(B^3\)
functions, and each bucket 's list  has \(O_d(k/B)\) records. Thus, all lists
together have \(O_d(k)\) records, using \(O_d(k\log k)\) variables to
store their identifiers. The evaluation tables have \(O_d(k^2/B)\)
entries in total. Including the auxiliary variables used to construct
them gives \(O_d((k^2/B)\log^2 k)\) variables. Function selection and
the associated checks use another \(O_d(kB^3)\) variables. Increasing
 the number of buckets reduces the   former cost but increases the latter.
We take \(B=\Theta(k^{1/4}\sqrt{\log k})\) to balance these two costs,
giving the \(O_d(k^{7/4}\log^{3/2}k)\)  variable bound in
\Cref{lem:deterministic-cnf}.

\subsection{Step 3: Reducing Back to \(d\)-Hitting Set}
\label{sec:variable-pair}

We next reduce the CNF formulas supplied by
\Cref{lem:random-cnf,lem:deterministic-cnf} back to \HS{}.
We use the variable-pair reduction for bounded-width
CNF~\cite{CyganEtAl2015} to turn formulas with few variables into
instances with few elements.

\begin{lemma}\label{lem:variable-pair}
Let \(\Psi\) be a CNF formula with \(N\) variables and \(M\) clauses,
each of width at most \(d\), where \(d\geq3\). In time linear in the
encoding length of \(\Psi\), one can construct an equivalent \(d\)-Hitting Set instance with exactly \(2N\) elements, at most \(N+M\) sets,
and parameter \(N\).
\end{lemma}

\begin{proof}
For every Boolean variable \(u\), introduce two elements
\(u^{(0)}\) and \(u^{(1)}\), and add the set
\[
\{u^{(0)},u^{(1)}\}.
\]
 The elements \(u^{(0)}\) and \(u^{(1)}\) represent the assignments
\(u=0\) and \(u=1\), respectively.
 For every clause, add a set obtained by replacing each positive literal
\(u\) with \(u^{(1)}\) and each negative literal \(\neg u\) with
\(u^{(0)}\).

Given a satisfying assignment, choose \(u^{(1)}\) when \(u\) is true
and \(u^{(0)}\) when \(u\) is false. This selects one element from every
variable pair and hits every clause set.

Conversely, the \(N\) variable-pair sets are pairwise disjoint. A
hitting set of size at most \(N\) must therefore contain exactly one
element from every pair. These choices define a truth assignment.
Every clause set is hit only if at least one of its literals is true
under this assignment, so the formula is satisfied.

The construction introduces two elements and one pair set for each
variable, together with at most one set for each clause. Pair sets have
size two, and clause sets have size at most \(d\). The construction
scans the formula once and therefore runs in linear time.
\end{proof}

\subsection{Proofs of the Main Theorems}

Now we are ready to prove our main results \Cref{thm:intro-random-two-level,thm:intro-deterministic}. 

\begin{proof}[Proof of \Cref{thm:intro-random-two-level,thm:intro-deterministic}]
First apply the preprocessing of \Cref{sec:framework-preprocessing}.
For the   remaining  instance, construct a CNF
formula using \Cref{lem:random-cnf} for the randomized algorithm or
\Cref{lem:deterministic-cnf} for the deterministic algorithm, and then
apply \Cref{lem:variable-pair} to reduce it back to \HS{}.

Since \Cref{lem:variable-pair} preserves equivalence, the randomized
construction   has the same  one-sided error and success probability of at
least \(3/4\)   as in  \Cref{lem:random-cnf}, while the deterministic
construction preserves equivalence with the input.

By \Cref{lem:variable-pair}, the output has \(2N\) elements, at most
\(N+M\) sets, and parameter \(N\). Substituting the bounds from
\Cref{lem:random-cnf,lem:deterministic-cnf}, respectively, gives the
element bounds, set bounds, and output parameters in the two main
theorems. Combining the running times of the steps, the randomized
algorithm runs in expected polynomial time and the deterministic
algorithm runs in polynomial time.
\end{proof}

\section{Preliminaries}\label{sec:preliminaries}


Fix a constant \(d\geq3\) . A set system is a pair \(\mathcal H=(U,\mathcal S)\), where \(U\) is a finite universe and \(\mathcal S\subseteq 2^U\).
We let $n = |U|$ and $m = |\mathcal{S}|$.
 An instance of \HS{} consists of such a set system  with \(|A|\leq d\)
for every \(A\in\mathcal S\),  and a nonnegative integer \(k\).
It is a YES-instance if there exists a set \(Z\subseteq U\) such that \(|Z|\leq k\) and \(Z\cap A\neq\varnothing\) for every \(A\in\mathcal S\).
Such a set \(Z\) is called a hitting set.

For a nonnegative integer \(r\), let
\([r]=\{1,\ldots,r\}\).
The notation \(O_d(\cdot)\) hides constants that may depend on the fixed
  constant  \(d\), but not on \(k\) or the input size.
 All logarithms are to base two.

\subsection{Kernelization}

A kernelization, or simply a kernel, is a polynomial-time algorithm that maps an instance \((I,k)\) to an instance \((I',k')\) of the same problem such that \((I,k)\)  is a YES-instance if and only if \((I',k')\)  is a YES-instance and $|I'|+k' \le f(k)$, where $f(k)$ is a computable function of \(k\) .

In this paper, a one-sided randomized kernel is allowed expected
polynomial running time and has the same output-size guarantee.
  It never maps a  NO-instance to a   YES-instance , while it maps each
YES-instance to a YES-instance with a stated success probability.

\subsection{{Subhypergraph Kernels}}\label{sec:induced-kernels}

{A subhypergraph kernel deletes elements and whole sets from
the input. Formally, its output \((U',\mathcal S',k')\) satisfies
\[
U'\subseteq U,\qquad
\mathcal S'\subseteq\{A\in\mathcal S:A\subseteq U'\}.
\]
The output must be equivalent to the input and satisfy the usual
kernel-size bound.
 }

{  We prove the  following lower bound   for subhypergraph kernels  using the sparsification lower bound of Dell and van Melkebeek \cite{DellVanMelkebeek2014}.}

\begin{proposition}\label{prop:induced-lower-bound}
{For every fixed integer \(d\geq3\) and every fixed
\(0<\varepsilon<d-1\), \(d\)-Hitting Set has no deterministic
kernel with \(O(k^{d-1-\varepsilon})\) elements whose output is a
subhypergraph of the input, unless
\(\mathsf{coNP}\subseteq\mathsf{NP/poly}\).}
\end{proposition}

\begin{proof}
{The idea is to give each set its own distinct element, so that
  an output with many sets must also contain  many elements.

Consider a Hitting Set instance \((U,\mathcal S,k)\), where
\(|U|=n\), every set has size exactly \(d-1\), and \(k\leq n\).
For each \(A\in\mathcal S\), add a new element \(p_A\) that occurs
only in this set, and replace \(A\) by \(A\cup\{p_A\}\).
This gives an equivalent \(d\)-Hitting Set instance with the same
parameter. Indeed, every original hitting set remains valid.
Conversely, replacing each selected \(p_A\) by any element of \(A\)
gives an original hitting set of no greater size.

Suppose that the kernel in the statement exists, and apply it to
this expanded instance. Let \(N=O(k^{d-1-\varepsilon})\) be the
number of output elements.
Since the output is a subhypergraph, every   output  set still
contains its distinct element \(p_A\). There are therefore at most
\(N\) output sets.
After numbering the output elements from \(1\) to \(N\), each set
can be recorded using \(O_d(\log N)\) bits.
The output parameter can be capped at \(N\), so the entire output
requires only \(O_d(N\log N)\) bits.

We have thus compressed the original instance in polynomial time to
\[
O_d\!\left(k^{d-1-\varepsilon}\log k\right)
\subseteq
O_d\!\left(n^{d-1-\varepsilon/2}\right)
\]
bits, where we use \(k\leq n\).
However, Dell and van Melkebeek
\cite{DellVanMelkebeek2014} show that Hitting Set
instances on \(n\) elements, with every set of size exactly \(d-1\),
cannot be compressed in deterministic polynomial time into
\(O(n^{d-1-\delta})\) bits while preserving the YES/NO answer,
for any fixed \(\delta>0\), unless
\(\mathsf{coNP}\subseteq\mathsf{NP/poly}\).
Taking \(\delta=\varepsilon/2\) gives the contradiction.}
\end{proof}

{The kernels of Abu-Khzam \cite{AbuKhzam2010} and Liu and Xiao \cite{LiuXiao2025} are not subhypergraph kernels, since their reductions may replace input sets by smaller sets that are not present in the input.}

 \subsection{CNF Formulas and Cardinality Constraints}

A literal is a Boolean variable or its negation. A clause is a disjunction of literals. A formula is in conjunctive normal form (CNF) if it is a conjunction of clauses.
The width of a clause is its number of literals. A $d$-CNF formula is a 
CNF formula such that the width of each clause is at   most  $d$.

  For a collection of Boolean variables, we often need to restrict how
many take the value \(1\), for example by requiring exactly one variable
or at most \(k\) variables to be true. These cardinality constraints can
be written as equalities or inequalities on sums of variables. The
following lemma converts them into CNF formulas using only linearly many
auxiliary variables and clauses, each containing  at most three  literals.
 \begin{lemma}[Linear-size cardinality encoding]\label{lem:counters}
Let \(x_1,\ldots,x_N\) be Boolean variables, let \(b\) be a
nonnegative integer, and let \(\bowtie\in\{=,\leq\}\). The constraint
\[
\sum_{i=1}^{N}x_i\mathrel{\bowtie}b
\]
can be encoded in polynomial time by a CNF formula such that an
assignment to \(x_1,\ldots,x_N\) extends to a satisfying assignment of
the formula if and only if it satisfies the constraint. The encoding
has \(O(N)\) auxiliary variables and \(O(N)\) clauses, each of
width at most three.
\end{lemma}

\begin{proof}
The cases \(N=0\) and \(b>N\) can be represented directly by a
constant true or false formula. Otherwise,   we use a standard parallel binary counter~\cite{Sinz2005},
implemented as a  balanced tree of binary addition circuits, to compute
 the sum of \(x_1,\ldots,x_N\). At level
\(j\), there are \(O(N/2^{j+1})\) additions on numbers represented by
\(O(j+1)\) bits. Hence the entire tree uses
\[
O\!\left(
\sum_{j\geq0}\frac{N(j+1)}{2^{j+1}}
\right)
=O(N)
\]
gates. Comparing the resulting sum with the fixed value \(b\) requires
only \(O(\log(N+1))\) additional gates.

Express the circuit using binary AND, OR, XOR, and NOT gates, and apply
the Tseitin encoding to each gate.  This encoding introduces a variable for each gate output and clauses enforcing the relation between that output and the gate inputs.  Each gate introduces constantly many
variables and clauses of width at most three. A unit clause requires the
comparison output to be true. The resulting CNF therefore has the
claimed satisfying assignments and size.
\end{proof}

 \section{Random Hashing and the Randomized Encoding}
\label{sec:randomized-kernels}

  This section proves the randomized encoding result in
\Cref{lem:random-cnf} using two-level hashing. After applying the known
kernelization in \Cref{lem:reduced-instance}, the universe \(U\) may still
contain \(O_d(k^{d-1})\) elements. Assigning a Boolean variable to every
element would therefore require \(|U|\) variables. A hitting set, however,
uses at most \(k\) elements, which suggests recording the selected elements
instead. Recall   that \(X\) is the union of the sets in a maximal
pairwise disjoint subfamily of \(\mathcal S\), \(W=U\setminus X\), and
 \(|X|\leq dk\). We represent the
elements of \(X\) directly by individual variables, and use two-level
hashing and a shared table to represent the selection of elements from
\(W\).

The first-level hash distributes the elements of \(W\) among buckets, and
the second-level hash   is used to separate  the selected elements within each bucket.
Together, the two hash values   identify a subbucket. We arrange these
subbuckets in a shared table and use Boolean variables to record a
selected element's address at each corresponding position. Different
assignments to these variables can record different elements of the
same subbucket, one at a time. Elements selected together therefore
need distinct subbuckets.  Below, we first generate the hashes, then define the table
variables and impose CNF constraints requiring at most \(k\) selected
elements to hit every input set, and finally prove the correctness and
size bounds of the resulting \(d\)-CNF formula.

\subsection{Two-Level Random Hashing}
\label{sec:random-two-level}

We first  use a hash function \(\widetilde h\colon W\to[B]\) to
 distribute the elements of \(W\) among \(B\) buckets  . Within each first-level bucket , we use a   second hash function to partition
its elements into \(q_2\) subbuckets belonging to that bucket. We provide  \(t\) candidate   functions
\(g_1,\ldots,g_t\colon W\to[q_2]\), forming a family \(\mathcal H_2\) .
All first-level buckets share   this family, but each bucket may choose a
different candidate .

  Each first-level bucket chooses one candidate function. We describe  these
choices by a map \(\sigma\colon[B]\to[t]\), where   \(\sigma(c)\) is the
index of the function chosen by  bucket \(c\)  . For each element \(v\in W\),
let \(c=\widetilde h(v)\) be its  first-level bucket  . The chosen function
\(g_{\sigma(c)}\) sends  \(v\)   to subbucket \(g_{\sigma(c)}(v)\). We record
its resulting location by the pair
\begin{equation}\label{eq:random-two-level-map}
H_\sigma(v)=\bigl(c,g_{\sigma(c)}(v)\bigr).
\end{equation}
Our goal is to choose these functions so that  elements selected together
  occupy distinct locations .

\begin{definition}[Two-level separation]\label{def:random-two-level-separation}
 Let \(T\subseteq W\) be a set of elements that we want to place in
distinct subbuckets.  The pair \(\mathcal F=(\widetilde h,\mathcal H_2)\)
separates   \(T\)  if there exists a map \(\sigma\colon[B]\to[t]\)
such that \(H_\sigma\) is injective on \(T\). Equivalently, for each
first-level bucket, some function in \(\mathcal H_2\) sends the elements
of \(T\) in that bucket to distinct subbuckets.
\end{definition}

To construct \(\mathcal F\), we   independently generate the first-level
function  \(\widetilde h\) and  the second-level candidates
 \(g_1,\ldots,g_t\)  by assigning every element of \(W\) an independent
uniform value in the corresponding range   for each function. A solution
contains  at most \(k\) elements   of \(W\), so we consider a fixed set
\(T\subseteq W\) with \(|T|\leq k\).

We first control how many  elements of \(T\)   fall into  each first-level
bucket.   The expected number is  at most \(k/B\)  . For a suitable bound
\[
b=O\!\left(\frac{k}{B}+\log k\right),
\]
all buckets contain  at most \(b\)   elements of \(T\) with high probability.
We then take \(q_2=\Theta(b^2)\) subbuckets per bucket, so that a single
random  second-level   function separates these elements  with constant probability.

The total number of subbuckets is  \(Bq_2=\Theta(Bb^2)\)  . Using fewer first-level
buckets increases the number of elements to separate within each bucket,
whereas using more buckets increases the number of table rows. We balance
these contributions by taking
\[
B=\Theta(k/\log k).
\]
This gives \(b=\Theta(\log k)\) ,   \(q_2=\Theta(\log^2 k)\), and
\(O(k\log k)\) subbuckets.

Finally, we sample  \(t=\Theta(\log k)\) candidate second-level functions  .
These provide  a constant probability that every  first-level  bucket has a
  candidate separating its elements of \(T\) . The following lemma gives
the  parameters and success probability.

\begin{lemma}\label{lem:random-two-level-distribution}
Let \(W\) be a finite set. There is an expected polynomial-time algorithm
that samples a pair \(\mathcal F=(\widetilde h,\mathcal H_2)\) with
\[
B=O(k/\log k),\qquad q_2=O(\log^2 k),\qquad t=O(\log k).
\]
For every fixed set \(T\subseteq W\) of size at most \(k\), the pair
\(\mathcal F\) separates \(T\) with probability at least \(3/4\).
\end{lemma}

\begin{proof}
Set \(\lambda=1+\lceil\log_2 k\rceil\)  , and take
\[
B=\left\lceil\frac{k}{\lambda}\right\rceil,
\qquad b=\min\{k,8\lambda\}.
\] 
Sample  the first-level function  \(\widetilde h\colon W\to[B]\)  uniformly
at random. Fix \(T\subseteq W\) with \(|T|\leq k\)  . For each bucket
\(c\in[B]\), write
\[
T_c=T\cap\widetilde h^{-1}(c)
\]
for the elements of \(T\) in that bucket .

We first   show that, with probability at least \(7/8\), every first-level
bucket contains at most \(b\) elements of \(T\) . If
\(|T|\leq b\), every \(|T_c|\) is at most \(b\). Otherwise,
\(b=8\lambda<k\), and a union bound over all \(B\) buckets and all
\((b+1)\)-element subsets of \(T\) gives
\begin{equation}\label{eq:random-load-union-bound}
\begin{aligned}
\Pr\!\left[\max_{c\in[B]}|T_c|>b\right]
&\leq B\binom{|T|}{b+1}B^{-(b+1)}\\
&\leq B\left(\frac{ek}{(b+1)B}\right)^{b+1}\\
&\leq 2^{\lambda-b-1}\leq\frac{1}{8}.
\end{aligned}
\end{equation}
Here \(B\geq k/\lambda\) and \(b=8\lambda\) make the parenthesized
ratio less than \(1/2\), while \(B\leq k<2^\lambda\).   This proves the claimed bound.

We now consider the second level. Take \(q_2=4b^2\), and suppose the
 first-level   function has placed  at most \(b\) elements of \(T\)   in every bucket. Under a uniformly random
function \(g\colon W\to[q_2]\) , each pair of distinct elements in
\(T_c\) collides with probability \(1/q_2\).   A  union bound over pairs   gives
\[
\Pr[g\text{ is not injective on }T_c]
\leq\frac{\binom{|T_c|}{2}}{q_2}
\leq\frac{\binom b2}{4b^2}<\frac{1}{8}.
\] 
  We therefore sample
\[
t=\lceil\log_8(8B)\rceil
\]
candidate functions \(g_1,\ldots,g_t\), independently of one another
and of \(\widetilde h\).

All  \(t\) candidates fail on  a fixed  \(T_c\) with probability at most \(8^{-t}\). A union bound
over the \(B\) buckets then gives
\[
\Pr[\text{some bucket has no injective candidate on }T_c]
\leq B 8^{-t}\leq\frac{1}{8}.
\]

When both   steps succeed , each bucket \(c\) has a candidate
\(g_{\sigma(c)}\) that is injective on \(T_c\), so the resulting
\(H_\sigma\) is injective on \(T\). The overall success probability is
at least \((7/8)^2=49/64>3/4\).

Since \(\lambda=\Theta(\log k)\), the parameters satisfy
\(B=O(k/\log k)\), \(q_2=O(\log^2 k)\), and \(t=O(\log k)\).
Sampling requires \(O(|W|(1+t))\) uniform draws from finite ranges  .
Rejection  sampling generates each value in an expected constant number of
trials using \(O(\log k)\) random bits per trial, so the total sampling
time is expected polynomial.
\end{proof}

The algorithm samples \(\mathcal F\) without knowing a solution. We now
fix this sample and use Boolean variables in the CNF formula to choose
\(\sigma\) and represent the selected elements.

\subsection{The Shared Table}
\label{sec:random-shared-table}

 For each \(x\in X\), we introduce a Boolean variable \(s_x\) indicating
whether \(x\) is selected. We now describe how to represent the selected
elements of \(W\).

 Two-level hashing aims to place elements selected together in different
subbuckets. We use   a  table with \(B\) rows  and \(q_2\) columns   to record
the selected elements, with position \((c,j)\) corresponding to subbucket
\(j\) of first-level bucket \(c\) . Once \(\sigma\) is chosen, an element
\(v\in W\) uses row \(c=\widetilde h(v)\) and column
\(j=g_{\sigma(c)}(v)\).   Since a subbucket may contain  several elements of
\(W\) , we assign each \(v\in W\) a distinct binary   identifier, called its
address, \(\gamma(v)\in\{0,1\}^{\ell}\),  where
\(\ell=\lceil\log_2(|W|+1)\rceil\)  . To  record a selected element  , we store
its address at the position corresponding to its subbucket .

  For each subbucket  \((c,j)\)  , introduce  a Boolean activation variable
\(a_{c,j}\) and \(\ell\) Boolean address variables
\(z_{c,j,1},\ldots,z_{c,j,\ell}\). If \(a_{c,j}=0\),   these variables
represent  no element. If \(a_{c,j}=1\) and the stored address matches
\(\gamma(v)\) for an element \(v\)   in subbucket  \((c,j)\), then   they represent  \(v\).

To see how   the same variables are  reused across candidate elements,
suppose   bucket \(c\) chooses \(g_1\), and two elements \(u,v\) in that
bucket  fall into the same   subbucket  \((c,j)\) . Setting \(a_{c,j}=1\) and
the address bits to \(\gamma(u)\) represents \(u\), whereas storing
\(\gamma(v)\) represents \(v\). These are simply different assignments
to the same variables, so   they  can represent either element when at most
one of them is selected. If both \(u\) and \(v\) belong to the solution
and \(\mathcal F\) separates the solution's elements in \(W\), bucket
\(c\) can instead choose a function \(g_r\) that maps them to distinct
  subbuckets , allowing both to be represented at the same time.

  We also introduce a Boolean variable  \(w_{c,r}\) for each \(c\in[B]\)
and \(r\in[t]\) ,  indicating whether bucket \(c\) uses \(g_r\). Notice
that the table variables \(a_{c,j}\) and \(z_{c,j,\beta}\) carry no
index \(r\)  . All  candidate functions in \(\mathcal H_2\) share this
single table, while the chosen functions determine which   subbucket each
element belongs to . Table~\ref{tab:random-main-variables} summarizes the main variables.

\begin{table}[ht]
\centering
\small
\begin{tabular}{@{}lll@{}}
\toprule
Variable & Index range & Meaning\\
\midrule
\(s_x\) & \(x\in X\) & Whether \(x\) is selected\\
\(w_{c,r}\) & \((c,r)\in[B]\times[t]\) & Whether bucket \(c\) uses \(g_r\)\\
\(a_{c,j}\) & \((c,j)\in[B]\times[q_2]\) & Whether   subbucket  \((c,j)\)  's variables are  active\\
\(z_{c,j,\beta}\) & \((c,j,\beta)\in[B]\times[q_2]\times[\ell]\) & Bit \(\beta\) of the address stored in \((c,j)\)\\
\bottomrule
\end{tabular}
\caption{The main variables of the encoding.}
\label{tab:random-main-variables}
\end{table}

Given \(\sigma\), an element \(v\in W\) is represented precisely when
  the variables for its subbucket  \((c,j)=H_\sigma(v)\)   satisfy
\(a_{c,j}=1\) and store  \(\gamma(v)\). We
write this test as
\begin{equation}\label{eq:random-entry-test}
D_{\sigma,v}=a_{c,j}\land
\bigwedge_{\substack{\beta\in[\ell]\\\gamma(v)_\beta=1}}z_{c,j,\beta}
\land
\bigwedge_{\substack{\beta\in[\ell]\\\gamma(v)_\beta=0}}\neg z_{c,j,\beta}.
\end{equation}
Here , \(\gamma(v)_\beta\) denotes the \(\beta\)-th bit of the address
\(\gamma(v)\). The expression  \(D_{\sigma,v}\) abbreviates a conjunction of \(1+\ell\) literals
and is not a new variable. The following bound allows us to control the
number of represented elements in \(W\) by limiting the number of
  activation variables set to one .

\begin{lemma}\label{lem:random-entry-count}
Fix \(\sigma\) and an assignment to the table variables.   The variables
for each subbucket represent  at most one element of \(W\). Consequently,
\[
\bigl|\{v\in W:D_{\sigma,v}=1\}\bigr|
\leq\sum_{c=1}^{B}\sum_{j=1}^{q_2}a_{c,j}.
\]
\end{lemma}

\begin{proof}
Elements   in the same subbucket  have distinct addresses. A fixed
assignment to   that subbucket's address variables  matches at most one of
them, and  represents that element only when its activation variable is one.
\end{proof}

\subsection{The CNF Encoding}
\label{sec:random-cnf}

We now constrain these variables so that they select one function per
first-level bucket, represent at most \(k\) elements, and hit every input
set. The first two requirements are
\begin{align}
\sum_{r=1}^{t}w_{c,r}&=1\qquad(c\in[B]),\label{eq:random-one-choice}\\
\sum_{x\in X}s_x+\sum_{c=1}^{B}\sum_{j=1}^{q_2}a_{c,j}&\leq k.
\label{eq:random-cardinality}
\end{align}
Equation~\eqref{eq:random-one-choice}   requires each bucket \(c\) to
choose exactly one function \(g_r\), defining \(\sigma(c)=r\) . In
\eqref{eq:random-cardinality}, the first sum counts the selected elements
of \(X\), and \Cref{lem:random-entry-count} bounds the number of
represented elements of \(W\) by the second sum. By \Cref{lem:counters},
both constraints admit  CNF encodings of width at most three , with
numbers of auxiliary variables and clauses linear in the number of
variables appearing in the sums .

The remaining requirement is to hit every input set. For a fixed
\(\sigma\), hitting \(A\in\mathcal S\) would simply mean
\[
\bigvee_{x\in A\cap X}s_x\ \vee\!\bigvee_{v\in A\cap W}D_{\sigma,v}.
\]
Since \(\sigma\) is chosen by the variables \(w_{c,r}\), we must encode
this condition across the possible choices of \(\sigma\). For a single
set \(A\), however, only the choices made in the buckets met by
\(A\cap W\) affect the condition.   Let \(C_A\) be the set of indices
of these first-level buckets, so
 \(C_A=\{\widetilde h(v):v\in A\cap W\}\).   We enumerate the  second-level
 function  choices for these buckets   and write \(\rho\colon C_A\to[t]\)
for one such local choice. Since \(|C_A|\leq|A\cap W|\leq d-1\),
there are at most \(t^{d-1}\) such choices.

For each local   choice  \(\rho\), we require  the   corresponding hitting
check to hold only when these buckets actually choose the functions
specified by \(\rho\). We express this condition using auxiliary
variables while keeping the final  clause width at most \(d\)  .

For each  \((c,r,j)\in[B]\times[t]\times[q_2]\), introduce  a variable
 \(p_{c,r,j}\)   satisfying
 \begin{equation}\label{eq:random-gated-activation}
p_{c,r,j}\leftrightarrow(\neg w_{c,r}\vee a_{c,j}).
\end{equation}
  When bucket \(c\) chooses \(g_r\), this variable checks whether
\(a_{c,j}=1\). Otherwise, it is true.

We treat the address bits in the same way. For each  \(\beta\in[\ell]\),
introduce   \(p^1_{c,r,j,\beta}\) and \(p^0_{c,r,j,\beta}\) satisfying
\begin{equation}\label{eq:random-gated-address}
\begin{aligned}
p^1_{c,r,j,\beta}&\leftrightarrow(\neg w_{c,r}\vee z_{c,j,\beta}),\\
p^0_{c,r,j,\beta}&\leftrightarrow(\neg w_{c,r}\vee\neg z_{c,j,\beta}).
\end{aligned}
\end{equation} 
 When \(g_r\) is chosen, these variables check whether address bit
\(\beta\) is \(1\) or \(0\), respectively. Otherwise, both are true.
 Each equivalence uses constantly many CNF clauses of width at most three.

 For \(r\in[t]\) and \(v\in W\), put \(c=\widetilde h(v)\) and
\(j=g_r(v)\)  . The combined check is
 \begin{equation}\label{eq:random-gated-test}
\widehat D_{r,v}=p_{c,r,j}\land
\bigwedge_{\beta=1}^{\ell}p^{\gamma(v)_\beta}_{c,r,j,\beta}.
\end{equation}
  The bit \(\gamma(v)_\beta\) determines whether the check uses
\(p^0_{c,r,j,\beta}\) or \(p^1_{c,r,j,\beta}\).
It checks activation and address matching when bucket \(c\) chooses
\(g_r\), and is automatically true otherwise .

For every \(A\in\mathcal S\) and every local map
\(\rho\colon C_A\to[t]\), we now require
\begin{equation}\label{eq:random-local-hitting}
F_{A,\rho}=\bigvee_{x\in A\cap X}s_x\ \vee\!
\bigvee_{v\in A\cap W}\widehat D_{\rho(\widetilde h(v)),v}.
\end{equation}
If \(\rho\) differs from the selected map \(\sigma\) at some bucket in
\(C_A\), any element of \(A\cap W\) in that bucket contributes a true
conditional check, making \(F_{A,\rho}\) automatically true. Conversely,
the unique local map \(\rho\) that agrees with \(\sigma\) on all of
\(C_A\) reduces \(F_{A,\rho}\) to the actual hitting condition for
\(A\). Thus unselected local choices impose no restriction, while the
selected choice requires an element of \(A\) to be represented.

It remains to expand \(F_{A,\rho}\) into CNF. Put \(L=1+\ell\). In
\eqref{eq:random-local-hitting}, each \(x\in A\cap X\) contributes a
single literal \(s_x\), and each \(v\in A\cap W\) contributes a
conjunction \(\widehat D_{\rho(\widetilde h(v)),v}\) of \(L\) literals.
By distributivity, \(F_{A,\rho}\) is equivalent to the conjunction of all
clauses obtained by keeping all \(s_x\) for \(x\in A\cap X\) and
choosing one literal from \(\widehat D_{\rho(\widetilde h(v)),v}\) for
each \(v\in A\cap W\). Every resulting clause   takes one literal from
the term corresponding to each  element of \(A\), so its width is at most \(|A|\leq d\), and the number
of clauses is at most
\begin{equation}\label{eq:random-expansion-size}
L^{|A\cap W|}\leq L^{d-1}.
\end{equation}
For example, consider a set \(A=\{x,u,v\}\) with \(x\in X\) and
\(u,v\in W\), where \(u\) and \(v\) fall into the same first-level
bucket \(c\). For the local map assigning \(\rho(c)=r\), let
\(j_u=g_r(u)\) and \(j_v=g_r(v)\). The condition \(F_{A,\rho}\) is then
\(s_x\vee\widehat D_{r,u}\vee\widehat D_{r,v}\). Its CNF expansion keeps
\(s_x\) and picks one literal from \(\widehat D_{r,u}\) and one from
\(\widehat D_{r,v}\) in every combination: choosing both activation
checks gives the clause \(s_x\vee p_{c,r,j_u}\vee p_{c,r,j_v}\), while
the remaining clauses combine activation and address-bit checks. This
yields \(L^2\) clauses, each of width  at most  three. Longer  addresses increase
the number of clauses, not their width.

Let \(\Phi_{\mathcal F}\) be the conjunction of the cardinality encodings
of \eqref{eq:random-one-choice} and \eqref{eq:random-cardinality}, the
defining clauses of \eqref{eq:random-gated-activation} and
\eqref{eq:random-gated-address}, and the CNF expansion of \(F_{A,\rho}\)
for every \(A\in\mathcal S\) and \(\rho\colon C_A\to[t]\). Since
\(d\geq3\), every clause of \(\Phi_{\mathcal F}\) has width at most
\(d\).

\subsection{Correctness and the Formula Size}
\label{sec:random-correctness}

The following lemma, together with the separation probability in
\Cref{lem:random-two-level-distribution}, establishes the one-sided error
guarantee.

\begin{lemma}\label{lem:random-encoding-correctness}
The formula \(\Phi_{\mathcal F}\) is satisfiable if and only if there is
a hitting set \(Z\subseteq U\) of size at most \(k\) such that
\(\mathcal F\) separates \(Z\cap W\).
\end{lemma}

\begin{proof}
First suppose that such a hitting set \(Z\) exists, and choose
\(\sigma\) so that \(H_\sigma\) is injective on \(Z\cap W\). For each
bucket \(c\), set \(w_{c,\sigma(c)}=1\) and set all other selection
variables for bucket \(c\) to zero. Set \(s_x=1\) if and only if
\(x\in Z\cap X\). For each \(v\in Z\cap W\),   set the activation variable of subbucket
 \(H_\sigma(v)\)   to one and its address variables to  \(\gamma(v)\) .
These assignments do not conflict because \(H_\sigma\) is injective on
\(Z\cap W\). Set all remaining table variables to zero. Exactly
\(|Z\cap W|\)   activation variables are one , so the left-hand side of \eqref{eq:random-cardinality} equals
\(|Z\cap X|+|Z\cap W|=|Z|\leq k\). Assign the conditional variables
according to \eqref{eq:random-gated-activation} and
\eqref{eq:random-gated-address}, and extend the assignment to the
auxiliary variables of the cardinality constraints using
\Cref{lem:counters}.

Fix \(A\in\mathcal S\) and \(\rho\colon C_A\to[t]\). If \(\rho\)
differs from \(\sigma\) at some bucket \(c\in C_A\), pick an element
\(v\in A\cap W\) with \(\widetilde h(v)=c\)  . Then
 \(w_{c,\rho(c)}=0\), which makes \(\widehat D_{\rho(c),v}=1\) and
satisfies \(F_{A,\rho}\). If \(\rho\) agrees with \(\sigma\) on
\(C_A\), then \(\widehat D_{\rho(\widetilde h(v)),v}=D_{\sigma,v}\)
for every \(v\in A\cap W\). Since \(Z\) hits \(A\), some element of
\(A\cap Z\) makes either \(s_x=1\) (if it lies in \(X\)) or
\(D_{\sigma,v}=1\) (if it lies in \(W\)), satisfying \(F_{A,\rho}\).
Hence all hitting conditions and their equivalent CNF expansions are
satisfied.

Conversely, suppose that \(\Phi_{\mathcal F}\) has a satisfying
assignment. Equation~\eqref{eq:random-one-choice} selects a unique index
\(\sigma(c)\in[t]\) with \(w_{c,\sigma(c)}=1\) in each bucket
\(c\). Define
\[
Z_X=\{x\in X:s_x=1\},\qquad Z_W=\{v\in W:D_{\sigma,v}=1\}.
\]
By \Cref{lem:random-entry-count}, the elements of \(Z_W\) map to
distinct   subbuckets whose activation variables are one  under \(H_\sigma\). Thus \(H_\sigma\) is
injective on \(Z_W\) (so \(\mathcal F\) separates \(Z_W\)), and
\(|Z_W|\leq\sum_{c=1}^{B}\sum_{j=1}^{q_2}a_{c,j}\). Together with
\eqref{eq:random-cardinality} and \(X\cap W=\varnothing\), this gives
\(|Z_X\cup Z_W|=|Z_X|+|Z_W|\leq k\). Finally, for each
\(A\in\mathcal S\), consider the local map \(\rho=\sigma|_{C_A}\)  that agrees with the selected functions .
For this choice of \(\rho\), every conditional check
\(\widehat D_{\rho(\widetilde h(v)),v}\) in \(F_{A,\rho}\) has the
same truth value as \(D_{\sigma,v}\). Since \(F_{A,\rho}\) is
satisfied, \(A\) contains an element of \(Z_X\cup Z_W\). Therefore
\(Z=Z_X\cup Z_W\) is a hitting set of size at most \(k\) such that
\(\mathcal F\) separates \(Z\cap W\).
\end{proof}

We now bound the numbers of variables and clauses in
\(\Phi_{\mathcal F}\). Let \(N\) and \(M\) denote the total
numbers of variables and clauses in \(\Phi_{\mathcal F}\), including
all auxiliary variables. Recall that \(m=|\mathcal S|\) and \(L=1+\ell\).

\begin{lemma}\label{lem:random-formula-size}
The numbers of variables and clauses in \(\Phi_{\mathcal F}\) satisfy
\begin{equation}\label{eq:random-formula-size}
\begin{aligned}
N&=O_d\!\left(k+Btq_2L\right),\\
M&=O_d\!\left(mt^{d-1}L^{d-1}+N\right).
\end{aligned}
\end{equation}
\end{lemma}

\begin{proof}
The formula has \(|X|=O_d(k)\) direct variables, \(Bt\)
function-selection variables, and \(Bq_2L\) table variables. The
conditional checks in \eqref{eq:random-gated-activation} and
\eqref{eq:random-gated-address} introduce  one activation-check variable
and \(2\ell\) address-bit-check variables for each \((c,r,j)\), giving
 \(Btq_2(1+2\ell)=O(Btq_2L)\) variables  in total . By
\Cref{lem:counters}, the cardinality encodings add a number of auxiliary
variables and clauses linear in the lengths of their variable lists,
giving the bound on \(N\).

For each \(A\in\mathcal S\), there are at most \(t^{d-1}\) local maps
\(\rho\colon C_A\to[t]\), and by \eqref{eq:random-expansion-size}
each contributes at most \(L^{d-1}\) clauses. Thus the hitting
conditions contribute at most \(mt^{d-1}L^{d-1}\) clauses in total,
while the conditional-variable definitions and cardinality encodings
contribute \(O_d(N)\) additional clauses.
\end{proof}

\begin{proof}[Proof of \Cref{lem:random-cnf}]
For the preprocessed input, sample \(\mathcal F\) as in
\Cref{lem:random-two-level-distribution} and construct the \(d\)-CNF
formula \(\Phi_{\mathcal F}\).

If the input is a YES-instance, fix a hitting set \(Z\subseteq U\) of
size at most \(k\). By \Cref{lem:random-two-level-distribution}, the
sampled pair \(\mathcal F\) separates \(Z\cap W\) with probability at
least \(3/4\)  . When  this holds, \Cref{lem:random-encoding-correctness}
implies that \(\Phi_{\mathcal F}\) is satisfiable. Conversely, for every
sampled \(\mathcal F\), \Cref{lem:random-encoding-correctness} guarantees
that a satisfying assignment of \(\Phi_{\mathcal F}\) yields a hitting
set of size at most \(k\). Thus every NO-instance produces an unsatisfiable formula .

By \Cref{lem:reduced-instance}, we have \(|W|\leq n=O_d(k^{d-1})\) and
\(m=O_d(k^d)\), so \(L=1+\lceil\log_2(|W|+1)\rceil=O_d(\log k)\).
Substituting these bounds together with \(B=O(k/\log k)\),
\(q_2=O(\log^2 k)\), and \(t=O(\log k)\) from
\Cref{lem:random-two-level-distribution} into
 the variable and clause bounds in  \Cref{lem:random-formula-size} gives
\[
\begin{aligned}
N&=O_d\!\left(k+\frac{k}{\log k}\cdot\log k\cdot\log^2 k\cdot\log k\right)
=O_d(k\log^3 k),\\
M&=O_d\!\left(k^d\log^{2d-2}k+k\log^3 k\right)
=O_d(k^d\log^{2d-2}k).
\end{aligned}
\]
Finally, sampling \(\mathcal F\) takes expected polynomial time by
\Cref{lem:random-two-level-distribution}, and constructing
\(\Phi_{\mathcal F}\) takes time polynomial in \(k\) for every fixed \(d\).
\end{proof}

\section{Deterministic Hashing and the Deterministic Encoding}
\label{sec:det}

This section proves the deterministic encoding result in \Cref{lem:deterministic-cnf}.
As in the randomized encoding of \Cref{sec:randomized-kernels}, we start
from the common packing partition in \Cref{sec:packing},  where \(X\)
is the union of the sets in the packing and \(W=U\setminus X\).
We  represent the elements of \(X\) directly by individual Boolean
variables, and use hashing and shared tables to represent the selected
elements in \(W\).

In \Cref{sec:randomized-kernels}, each bucket chooses a second-level
function from a small number of random candidates to separate its selected
elements into distinct subbuckets.   A deterministic construction needs
candidates that cover every possible selection of elements. Since the
 conditional check variables are indexed by the candidate function,
  enumerating these candidates  multiplies this part of the variable count
by   their number. We therefore use  a small explicit hash family   to bound
 the number of selected elements in each bucket, and   record  their complete
identifiers  within that bucket.  Specifically, if this bound is \(b\),
we provide \(b\) storage positions per bucket. Each position has an
activation variable and identifier variables. Their assignment determines
whether the position is used and which selected element it stores.
Several selected elements in the same bucket can thus occupy different
positions without a second level of hashing to separate them.

  The remaining difficulty is checking whether  an input set   is hit.
We want many candidate elements to share the information used for this
check, rather than introducing  a match variable for  each element  .
We represent element identifiers as  low-degree   polynomials, substitute
the same set of distinct values into them, and build  a shared table  for
each bucket recording the results for its stored elements. Elements of
an input set query these tables to check for a match with a stored element .
  Different elements may give the same result at one value, whereas identical
elements match at every value. The bound on the number of points where
 distinct low-degree polynomials agree   lets us choose sufficiently many
values so that every unhit set fails a check. Thus, the shared tables
reduce the number of variables needed for the checks while ensuring
their correctness .

\subsection{A Hash Family with Bounded Bucket Load}
\label{sec:det-hash}

We again use \(B\) for the number of buckets and \(b\) for the
upper bound on the number of selected elements in each bucket.   We  construct
an explicit family of \(t=B^3\) candidate functions
\(h_1,\ldots,h_t\colon W\to[B]\) . For every set \(T\subseteq W\)
with \(|T|\leq k\) , at least one   function in this family  places at most
\(b=O_d(k/B)\) elements of \(T\) in every bucket. The CNF formula
chooses a function from this family to distribute all of \(W\) among
the buckets.

Increasing \(B\) reduces the number of elements each bucket needs to
hold, making the evaluation tables smaller, but also increases the
number of candidate functions. We ultimately take
\[
B=\Theta\!\left(k^{1/4}\sqrt{\log k}\right),
\]
with   the following number of storage positions in each bucket
 \[
b=\Theta_d\!\left(\frac{k^{3/4}}{\sqrt{\log k}}\right).
\]
This choice balances the variables needed for hash selection and the
evaluation tables, as calculated in \Cref{sec:det-bound}. We first give
the construction for a power of two \(B\) satisfying
\(k^{1/4}\leq B\leq\sqrt{k}\).

Let \(\mathbb F_B\) be the finite field of order \(B\), and let \(a\) be
the smallest positive integer with \(B^a\geq n+1\) , where \(n=|U|\) . Number the elements of
\(W\) from \(1\) to \(|W|\), and   write each number using \(a\) base-\(B\)
digits, padding with leading zeros. Fix a bijection between the digits
and \(\mathbb F_B\), and use these digits in order as coefficients of
a  distinct polynomial \(f_v\in\mathbb F_B[z]\) of degree less than \(a\)  .
Each candidate function evaluates \(f_v\) at two points and combines the
results linearly to obtain a bucket label. Specifically, enumerate  all
\(\beta_1,\beta_2,\lambda\in\mathbb F_B\) and define
\begin{equation}\label{eq:det-balanced-hashes}
h_{\beta_1,\beta_2,\lambda}(v)=f_v(\beta_1)+\lambda f_v(\beta_2).
\end{equation}
Identify the elements of \(\mathbb F_B\) with the bucket labels
\(c\in[B]\), list these \(t=B^3\) functions as \(h_1,\ldots,h_t\), and   provide
 \(b=\lceil(a+2)k/B\rceil\)  storage positions in each bucket . Since
\(n=O_d(k^{d-1})\) and \(B\geq k^{1/4}\), we have \(a=O_d(1)\), so
\(b=\Theta_d(k/B)\) and \(Bb=O_d(k)\).

\begin{lemma}\label{lem:det-balanced-hashes}
For every subset \(T\subseteq W\) with \(|T|\leq k\), at least one
function \(h_i\) in \eqref{eq:det-balanced-hashes} satisfies
\(|T\cap h_i^{-1}(c)|\leq b\) for every bucket \(c\in[B]\).
\end{lemma}

\begin{proof}
  A bucket containing too many elements creates many colliding pairs.
We first bound how many candidate functions make each pair collide,
then show that some function keeps every bucket within the bound .
For distinct \(u,v\in W\), the nonzero polynomial \(P=f_u-f_v\) has degree
at most \(a-1\) and thus has \(z\leq a-1\) roots in \(\mathbb F_B\). If
\(P(\beta_2)\neq0\), each \(\beta_1\) determines a unique \(\lambda\) for
which \(P(\beta_1)+\lambda P(\beta_2)=0\), giving \(B(B-z)\) triples. If
\(P(\beta_2)=0\), a collision requires \(P(\beta_1)=0\), and then all
\(B\) choices of \(\lambda\) work, giving \(Bz^2\) triples. Among all
\(t=B^3\) functions, the number of functions with \(h_i(u)=h_i(v)\) is
therefore
\[
B(B-z)+Bz^2\leq B^2+a^2B=t\left(\frac{1}{B}+\frac{a^2}{B^2}\right).
\]

Fix \(T\subseteq W\) with \(s=|T|\leq k\), and let
\(s_{i,c}=|T\cap h_i^{-1}(c)|\) be its size in bucket \(c\) under \(h_i\).
For each \(i\in[t]\), the sum \(\sum_{c=1}^B s_{i,c}^2\) counts   all
ordered pairs in the same bucket, including the  \(s\)   pairs \((v,v)\)
and the  ordered pairs of distinct elements  that collide under \(h_i\). Averaging over all \(t\) functions gives
\[
\begin{aligned}
\frac{1}{t}\sum_{i=1}^t\sum_{c=1}^B(s_{i,c}-s/B)^2
&=\frac{1}{t}\sum_{i=1}^t\sum_{c=1}^B s_{i,c}^2-\frac{s^2}{B}\\
&\leq s+s(s-1)\left(\frac{1}{B}+\frac{a^2}{B^2}\right)-\frac{s^2}{B}\\
&\leq k+\frac{a^2k^2}{B^2}
\leq(a^2+1)(k/B)^2,
\end{aligned}
\]
where the last inequality uses \(B^2\leq k\). Choosing a function \(h_i\)
whose   sum of squared deviations  is at most this average   gives
\(\lvert s_{i,c}-s/B\rvert\leq\sqrt{a^2+1}\,k/B\) for every bucket,
since each summand is nonnegative. Every bucket therefore  has size at most
\(s/B+\sqrt{a^2+1}\,(k/B)\leq(a+2)k/B\leq b\).
\end{proof}

\subsection{Bucket Records and Evaluation Tables}
\label{sec:det-tables}

For each \(x\in X\), we   use  the direct-selection variable \(s_x\). Each
bucket \(c\in[B]\) has \(b\)   storage positions, each holding one record.
Record  \(l\in[b]\) consists of an
activation bit \(a_{c,l}\) and  a block of variables \(u_{c,l}\) storing
 an element identifier , and the
record is active when \(a_{c,l}=1\). The cardinality constraint
\begin{equation}\label{eq:det-record-count}
\sum_{x\in X}s_x+\sum_{c=1}^B\sum_{l=1}^b a_{c,l}\leq k
\end{equation}
limits the total number of directly selected elements and active records
to at most \(k\).

The polynomials in \Cref{sec:det-hash} are used to construct the hash
family. To check the recorded identifiers, we use a separate polynomial
representation. Set \(r=2d+2\), and let   \(q\)  be the smallest
power of two satisfying \(q\geq4drb\) and \(q^r\geq n+1\).
  The condition \(q^r\geq n+1\) provides enough distinct identifiers, while
\(q\geq4drb\) provides enough evaluation points to ensure that an unhit
set fails a check. We represent each element \(v\in W\) by a distinct
vector \(\iota(v)\in\mathbb F_q^r\), and use the same \(r\)-coordinate
format for each record identifier \(u_{c,l}\) (which occupies
\(O_d(\log k)\) Boolean variables). For a vector
\(u=(u_0,\ldots,u_{r-1})\), write its associated polynomial as
\[
p_u(z)=\sum_{m=0}^{r-1}u_m z^m.
\]

Each bucket \(c\in[B]\) is equipped with a \(q\times q\) Boolean
evaluation table whose rows are indexed by evaluation points
\(\alpha\in\mathbb F_q\) and whose columns are indexed by possible values
\(\tau\in\mathbb F_q\).  Let \(Z_{\alpha,c,\tau}\) be the Boolean variable
in row \(\alpha\) and column \(\tau\) of bucket \(c\)'s table.
 Setting \(Z_{\alpha,c,\tau}=1\) certifies that
some active record in bucket \(c\) takes value \(\tau\) at \(\alpha\).
Multiple records taking the same value share the same table position.   To check an element \(v\) at an evaluation point \(\alpha\), we read
row \(\alpha\) and column \(p_{\iota(v)}(\alpha)\) in its bucket's table .

 Since \(B\leq\sqrt{k}\), we have \(b=\Omega_d(\sqrt{k})\), while \(r=2d+2\) gives \((n+1)^{1/r}=O_d(\sqrt{k})\). Thus, both lower bounds on \(q\) are \(O_d(b)\), and the smallest power of two satisfying them is \(\Theta_d(b)\).

 For a fixed bucket and evaluation point, we compute the  \(b\)  record
polynomials' values and pair each with its activation bit.   We want to
set to one the column corresponding to each value attained by an active
record. For each distinct such value, it suffices to place one pair with
activation bit one at the output position labelled by that value.
We append \(q-b\) pairs whose value and activation bit are both zero,
giving \(q\) pairs for the \(q\) columns.

We rearrange these pairs using  a network of switches. Each switch takes
two  pairs and, according to one control bit, either keeps them in place
or exchanges them.   Setting the control bits places the pairs at the
desired outputs . An output position is set to one   if and only if  the
arriving activation bit is one and the accompanying value matches   its
 label. This avoids comparing every table position against all \(b\)
records.

For example,   suppose three pairs have activation bit one and  values
\(1,1,3\)  . Moving one pair with value  \(1\) to position \(1\) and  the
pair with value  \(3\) to position \(3\) sets those two positions to one.
The   other pair with value  \(1\) goes to another position, such as \(2\),
and leaves it zero because its value does not match its destination.

 The following lemma   represents all such permutations  with few auxiliary
variables  while preserving the contents of the pairs .

\begin{lemma}\label{lem:record-permutation}
For \(p\geq2\) a power of two, a permutation of \(p\) records of \(w\)
bits each can be represented by \(O(pw\log p)\) variables and clauses of
width at most three: every satisfying assignment outputs a permutation of
the input records, and every permutation can be represented. The formula
can be constructed in polynomial time.
\end{lemma}

\begin{proof}
For two input records of \(w\) bits each, a switch applies the same
control bit across all \(w\) bit positions and uses \(O(w)\) Boolean gates.

We build the permutation network~\cite{Benes1964} recursively. Pair adjacent inputs and
pass each pair through a switch, sending one output to an upper subnetwork
of size \(p/2\) and the other to a lower subnetwork of size \(p/2\)  . Then
 pair the \(m\)-th outputs of the two subnetworks and pass them through a
final switch to output positions \(2m-1\) and \(2m\) (a network of size
one is a single wire). The number of switches satisfies
\(S(p)=2S(p/2)+p\), so \(S(p)=O(p\log p)\).

Every setting of the switches preserves the record contents. To realize a
given permutation,  we send one record from each input pair and each
target output pair through each subnetwork. To choose this division,
 temporarily distinguish the records by their input positions and form a
bipartite multigraph  . Its  left vertices are   pairs of input positions,
and its  right vertices are   pairs of target output positions. Each record
gives an edge  joining its input pair to its target output pair. Every vertex
has degree two, so the graph decomposes into even cycles. Alternating the
upper and lower subnetworks around each cycle sends one record from each
input pair and each output pair through each subnetwork. By induction, the
two subnetworks realize the induced permutations, and the final switches
place the records in the target order.

Encoding the Boolean gates of each switch   using the gate-to-CNF conversion in the proof of
\Cref{lem:counters} introduces constantly many variables and clauses of
width at most three per gate, giving the stated size.
\end{proof}

\begin{lemma}\label{lem:det-evaluation-table}
The evaluation tables for all buckets can be constructed using
\(O_d(Bbq\log^2 k)\) variables and clauses of width at most three so that
every satisfying assignment satisfies
\begin{equation}\label{eq:det-table-support}
Z_{\alpha,c,\tau}=1
\ \Longrightarrow\
\exists l\in[b]\quad a_{c,l}=1\ \text{and}\ p_{u_{c,l}}(\alpha)=\tau,
\end{equation}
and   any assignment to  the input records   can be extended to the auxiliary
variables to satisfy the formula and mark  all values attained by active
records  one in their table rows.
\end{lemma}

\begin{proof}
Fix a bucket \(c\) and an evaluation point \(\alpha\). Evaluate the \(b\)
record polynomials at \(\alpha\), keeping each value paired with its
activation bit. Append \(q-b\) pairs whose value and activation bit are
both zero, and pass the resulting \(q\) pairs through the permutation
network of \Cref{lem:record-permutation}. Label the \(q\) output positions
by the elements of \(\mathbb F_q\), and at output position \(\tau\)
require \(Z_{\alpha,c,\tau}=1\) if and only if the arriving activation
bit is one and the accompanying value equals \(\tau\).

Because the network only permutes these pairs,
\eqref{eq:det-table-support} holds immediately. Conversely, for each
value attained by an active record, choose one corresponding pair and
send it to the output position with the matching label. Placing the
remaining pairs in the remaining positions defines a permutation
realizable by \Cref{lem:record-permutation}, which marks every attained
value one.

For the size bound, let \(s=\log_2 q\). Evaluating a record polynomial at
the fixed point \(\alpha\) uses \(O_d(1)\) field additions and
multiplications by fixed constants: field addition is bitwise XOR, and
multiplication by a fixed constant is a linear map on \(s\) bits
implemented by \(O(s^2)\) XOR gates. Evaluating all \(b\) records   uses \(O_d(b\log^2q)\) gates. Permuting  the \(q\) pairs  uses \(O(q\log^2q)\) gates,  and checking their destinations uses   \(O(q\log q)\)  gates. Since \(q=\Theta_d(b)\),
  encoding these gates as CNF  uses \(O_d(b\log^2 k)\) variables and
clauses  per row . Multiplying by the
\(Bq\) rows across all buckets gives the stated bound.
\end{proof}

\subsection{The CNF Encoding}
\label{sec:det-cnf}

We introduce selection variables   \(w_1,\ldots,w_t\), where \(w_i=1\)
 means that hash function \(h_i\) is selected, and require
  \(\sum_{i=1}^t w_i=1\) . All hash functions share the same bucket records
and evaluation tables  . The  selected function determines which bucket each
original element uses.

For a fixed function \(h_i\), an input set \(A\in\mathcal S\), and an
evaluation point \(\alpha\in\mathbb F_q\), the natural check is
  \[
\neg w_i\ \vee\!\bigvee_{x\in A\cap X}s_x
\ \vee\!\bigvee_{v\in A\cap W}
Z_{\alpha,h_i(v),p_{\iota(v)}(\alpha)}.
\] 
When \(h_i\) is selected, this requires either a directly selected element
in \(A\cap X\) or   the value of  some \(v\in A\cap W\) at \(\alpha\)
 to be marked one in the evaluation table of bucket \(h_i(v)\) . The disjunction has \(|A|+1\) literals. To keep the clause
width at most \(d\), we   combine \(\neg w_i\) with one direct-selection
variable into a single literal. Every  input set meets \(X\)  , so we  choose an
element \(x_A\in A\cap X\) for each \(A\in\mathcal S\), and for every
\(i\in[t]\) and \(x\in X\) introduce an auxiliary variable satisfying
  \begin{equation}\label{eq:det-choice-guard}
p_{i,x}\leftrightarrow(\neg w_i\vee s_x).
\end{equation} 
We then require, for every \(i\in[t]\), \(A\in\mathcal S\), and
\(\alpha\in\mathbb F_q\), the clause
  \begin{equation}\label{eq:det-evaluation-clause}
p_{i,x_A}
\ \vee\!\bigvee_{x\in(A\cap X)\setminus\{x_A\}}s_x
\ \vee\!\bigvee_{v\in A\cap W}
Z_{\alpha,h_i(v),p_{\iota(v)}(\alpha)}.
\end{equation} 
Its width is \(|A|\leq d\), and all table positions are computable
directly from the input.

Let \(\Phi_{\mathrm{det}}\) consist of the clauses
\eqref{eq:det-evaluation-clause}, the evaluation-table constraints of
\Cref{lem:det-evaluation-table}, the defining clauses of
\eqref{eq:det-choice-guard}, and the cardinality encodings of
  \(\sum_{i=1}^t w_i=1\)  and \eqref{eq:det-record-count} from
\Cref{lem:counters}. Every clause of \(\Phi_{\mathrm{det}}\) has width at
most \(d\). Because the table variables occur only positively in
\eqref{eq:det-evaluation-clause}, failing to mark some attained values one
  can only make these clauses harder to satisfy.
On the other hand, \mbox{
\Cref{lem:det-evaluation-table} }\hskip0pt
ensures that all values
attained by active records can be marked  one.

\subsection{Correctness and the Formula Size}
\label{sec:det-bound}

\begin{lemma}\label{lem:det-encoding-correctness}
The formula \(\Phi_{\mathrm{det}}\) is satisfiable if and only if
\((U,\mathcal S,k)\) has a hitting set of size at most \(k\).
\end{lemma}

\begin{proof}
First suppose that \(Z\subseteq U\) is a hitting set with \(|Z|\leq k\).
By \Cref{lem:det-balanced-hashes}, choose a function \(h_i\) that maps at
most \(b\) elements of \(Z\cap W\) to each bucket. Set   \(w_i=1\)  (and all
other   \(w_{i'}=0\) ), set \(s_x=1\) if and only if \(x\in Z\cap X\), and
place the identifier \(\iota(v)\) of each \(v\in Z\cap W\) in a distinct
active record of bucket \(h_i(v)\), leaving all remaining records
inactive. Then \eqref{eq:det-record-count} holds. By
\Cref{lem:det-evaluation-table},   assign the evaluation-table and auxiliary
variables so that their constraints hold and  every value attained by these
active records is marked one.
Since \(Z\) hits every input set \(A\in\mathcal S\), the clause
\eqref{eq:det-evaluation-clause} is satisfied at every \(\alpha\) for the
selected \(h_i\), while for \(i'\neq i\) the literal   \(p_{i',x_A}\)  is
true. Thus \(\Phi_{\mathrm{det}}\) is satisfiable.

Conversely, fix a satisfying assignment of \(\Phi_{\mathrm{det}}\) and its
selected function \(h_i\). Let \(Z\) be the union of all \(x\in X\) with
\(s_x=1\) and all \(v\in W\) whose identifier \(\iota(v)\) appears in an
active record of bucket \(h_i(v)\). Because distinct elements have
distinct identifiers, each active record contributes at most one element,
so \eqref{eq:det-record-count} implies \(|Z|\leq k\).

Suppose for contradiction that some set \(A\in\mathcal S\) is not hit by
\(Z\). Then all its direct-selection literals are zero  .
Since \(w_i=1\), \eqref{eq:det-choice-guard} also gives
\(p_{i,x_A}=s_{x_A}=0\). For  each \(v\in A\cap W\), none of the active records in bucket \(h_i(v)\) equals
\(\iota(v)\). Their polynomials are therefore distinct from
\(p_{\iota(v)}\), so each agrees with \(p_{\iota(v)}\) at at most \(r-1\)
points. By \eqref{eq:det-table-support}, the table literal for \(v\) in
\eqref{eq:det-evaluation-clause} can be one at at most \(b(r-1)\)
evaluation points. Across the at most \(d\) elements of \(A\), these false
matches cover at most \(db(r-1)<q\) evaluation points. At any remaining
evaluation point \(\alpha\in\mathbb F_q\), every literal of
\eqref{eq:det-evaluation-clause} is zero, a contradiction. Thus \(Z\) hits
every input set.
\end{proof}

We now bound the total numbers of variables \(N\) and clauses \(M\) in
\(\Phi_{\mathrm{det}}\).

\begin{lemma}\label{lem:deterministic-kernel}
For a bucket count \(B\) satisfying \(k^{1/4}\leq B\leq\sqrt k\), the
numbers of variables and clauses in \(\Phi_{\mathrm{det}}\) satisfy
\begin{equation}\label{eq:det-variable-bound}
\begin{aligned}
N&=O_d\!\left(kB^3+\frac{k^2}{B}\log^2 k\right),\\
M&=O_d\!\left(k^{d+1}B^2+N\right).
\end{aligned}
\end{equation}
\end{lemma}

\begin{proof}
The bucket records use \(O_d(Bb\log k)=O_d(k\log k)\) variables. By
\Cref{lem:det-evaluation-table}, the evaluation tables and their
permutation networks use
\(O_d(Bbq\log^2 k)=O_d((k^2/B)\log^2 k)\) variables, since
\(Bb=O_d(k)\) and \(q=\Theta_d(b)=\Theta_d(k/B)\). The auxiliary variables   \(p_{i,x}\)
in \eqref{eq:det-choice-guard} number  \(|X|t=O_d(kB^3)\) .
The direct-selection variables \(s_x\) and function-selection  variables
  \(w_i\) number \(|X|+t\). Encoding the two cardinality constraints adds
\(O(|X|+Bb+t)\) variables, which also fit within the above bounds .
Summing these contributions gives the bound on \(N\).

  Equation~\eqref{eq:det-evaluation-clause} gives one hitting clause for each
input set, candidate function, and evaluation point, for a total of
 \(mtq\) clauses,   where \(m=|\mathcal S|\). All  other constraints contribute
\(O(N)\) clauses.
Since \(m=O_d(k^d)\), \(t=B^3\), and \(q=\Theta_d(k/B)\), we have
\(mtq=O_d(k^{d+1}B^2)\), proving the bound on \(M\).
\end{proof}

\begin{proof}[Proof of \Cref{lem:deterministic-cnf}]
To balance   \(kB^3\) and \((k^2/B)\log^2 k\) in the variable bound
 in \eqref{eq:det-variable-bound}, let \(B\) be
the smallest power of two satisfying \(B^4\geq k(\log_2 k)^2\), so
\(B=\Theta(k^{1/4}\sqrt{\log k})\).   For sufficiently large  \(k\)  ,
this choice satisfies \(k^{1/4}\leq B\leq\sqrt k\). For the remaining
finitely many values of \(k\) , enumerate all subsets of the reduced
universe of size at most \(k\), decide the instance, and return a fixed
satisfiable or unsatisfiable CNF formula accordingly.   Substituting
 \(B=\Theta(k^{1/4}\sqrt{\log k})\) into  the variable and clause bounds in  \Cref{lem:deterministic-kernel}
yields
\[
N=O_d(k^{7/4}\log^{3/2}k),\qquad
M=O_d(k^{d+3/2}\log k).
\]
By \Cref{lem:det-encoding-correctness}, the formula
\(\Phi_{\mathrm{det}}\) is satisfiable if and only if the input is a
YES-instance. Its clauses have width at most \(d\) by the construction
in \Cref{sec:det-cnf}.

All field sizes, hash families, permutation networks, and formulas are
polynomially bounded in \(k\). A binary field of order \(2^s\) with
\(s=O_d(\log k)\) can be constructed deterministically by enumerating
monic degree-\(s\) polynomials over \(\mathbb F_2\) and testing
divisibility by monic polynomials of degrees from \(1\) to
\(\lfloor s/2\rfloor\) until an
irreducible polynomial is found, which takes
\(2^{O(s)}\operatorname{poly}(s)=\operatorname{poly}(k)\) time. All
remaining steps consist of explicit finite-field arithmetic, Boolean gate
encodings, and clause generation, so the entire construction runs in
deterministic polynomial time.
\end{proof}

\section{Conclusion}\label{sec:conclusion}
 In this paper, we obtain   near-linear randomized and subquadratic
deterministic  element-kernels for \HS{}, resolving a longstanding
open problem in kernelization.

Our kernelization algorithms first transform an instance of \HS{} into
a CNF formula and then transform the resulting formula back into an
instance of \HS{}. When constructing the CNF formula, we re-encode the
universe rather than assign a separate Boolean variable to each input
element. This allows the same variables to participate in the
representations of multiple elements, while CNF clauses constrain their
joint assignments so that every satisfying assignment can still be
decoded into a valid hitting set. In the randomized kernel, this compact
representation is achieved through two-level hashing and short addresses.
In the deterministic kernel,   an explicit hash family distributes the
selected elements among buckets of bounded capacity. Each bucket stores
their identifiers, and shared polynomial evaluation tables check whether
each input set contains a stored element. Both constructions produce a
CNF formula of width at most \(d\). Finally, the variable-pair reduction
transforms the resulting formula back into an instance of
  \HS{} , with each CNF clause becoming an output set.
Consequently, the compressed variable representation substantially
reduces the number of output elements, albeit potentially at the cost of
increasing the number of sets and the value of parameter \(k'\).
 Finding a parameter-preserving reduction with the same element bounds is a direction for future research.
 We believe that this encoding framework may also prove useful for
obtaining improved element-kernels for other problems.

 \section*{Acknowledgements}

We thank Yuxi Liu, Kangyi Tian and Junqiang Peng for valuable discussions.
We also thank Saket Saurabh for his support and encouragement.
The work is supported by the National Natural Science Foundation of China under Grant Nos. 62372095 and 61972070.

\section*{Statement on the Use of Generative AI}\label{app:ai}

The authors had explored the   use of  hashing to compress solution
representations  before generative AI was used in this project. Our
central idea was to represent the elements of a small solution through
hash values, encode their choices in a bounded-width CNF formula, and
then reduce the formula back to \(d\)-Hitting Set.   Our  early
constructions were more complicated. In particular, incorporating the
choice of hash functions while keeping every clause of width at most
\(d\) required a rather involved encoding.

  In the earlier stages,  we worked interactively with OpenAI's
GPT-5.6 Sol model  to simplify this encoding. These discussions helped
us  find a simpler way to incorporate the hash choices while preserving
the clause-width bound \(d\).   We also used the model to explore
alternative constructions, examine encoding equivalences,  and   check
quantitative bounds.

We subsequently continued this work with GPT-6 Astra, revisiting
candidate improvements within the same framework and working through
their proofs. The model proposed construction ideas and proof
arguments, which we examined and revised through repeated discussions.
We placed particular emphasis on simplifying the techniques and
keeping them close to the original approach.  This process helped us
  develop the simplified deterministic construction with
\(O_d(k^{7/4}\log^{3/2}k)\) elements  presented in this paper.
 We also explored stronger randomized bounds, but chose to present the
\(O_d(k\log^3 k)\)-element randomized kernel because its construction
and proof are simpler.

  AI also assisted with the presentation of the  paper. We   discussed the
constructions and proofs in Chinese, organized the exposition according
to our understanding, and used AI to help draft and revise the English
text and explanatory figures.  The authors independently reviewed the
definitions, proofs, encoding equivalences, quantitative bounds, and
citations  in the manuscript. The authors take full
responsibility for the content and correctness of the paper.
 \bibliographystyle{alpha_surname}
  \bibliography{dHittingSet_20261031}

\newcommand{\etalchar}[1]{$^{#1}$}
\begin{thebibliography}{GKK{\etalchar{+}}15}

\bibitem[AK10]{AbuKhzam2010}
F.~N. Abu-Khzam.
\newblock A kernelization algorithm for {$d$-Hitting Set}.
\newblock {\em Journal of Computer and System Sciences}, 76(7):524--531, 2010.

\bibitem[AGJK20]{AgrawalEtAl2020Split}
Akanksha Agrawal, Sushmita Gupta, Pallavi Jain, and R.~Krithika.
\newblock Quadratic vertex kernel for split vertex deletion.
\newblock {\em Theoretical Computer Science}, 833:164--172, 2020.

\bibitem[AKK26]{AhnKimKwon2026}
Jungho Ahn, Jinha Kim, and O-joung Kwon.
\newblock Unified almost linear kernels for generalized covering and packing
  problems on nowhere dense classes.
\newblock {\em Journal of Computer and System Sciences}, 159:103782, 2026.

\bibitem[AYZ95]{AlonYusterZwick1995}
N.~Alon, R.~Yuster, and U.~Zwick.
\newblock Color-coding.
\newblock {\em Journal of the ACM}, 42(4):844--856, 1995.

\bibitem[BT20]{BannachTantau2018}
M.~Bannach and T.~Tantau.
\newblock Computing hitting set kernels by {$\mathrm{AC}^0$}-circuits.
\newblock {\em Theory of Computing Systems}, 64:374--399, 2020.

\bibitem[Ben64]{Benes1964}
V.~E. Bene{\v{s}}.
\newblock Optimal rearrangeable multistage connecting networks.
\newblock {\em Bell System Technical Journal}, 43(4):1641--1656, 1964.

\bibitem[BBTW23]{bessy2023kernelization}
St{\'e}phane Bessy, Marin Bougeret, Dimitrios~{M.} Thilikos, and Sebastian
  Wiederrecht.
\newblock Kernelization for graph packing problems via rainbow matching.
\newblock In {\em Proceedings of the 2023 Annual ACM-SIAM Symposium on Discrete
  Algorithms (SODA)}, pages 3654--3663. SIAM, 2023.
\newblock Full version:
  \href{https://arxiv.org/abs/2207.06874v4}{arXiv:2207.06874v4}.

\bibitem[BFL{\etalchar{+}}16]{BodlaenderEtAl2016Meta}
Hans~L. Bodlaender, Fedor~V. Fomin, Daniel Lokshtanov, Eelko Penninkx, Saket
  Saurabh, and Dimitrios~M. Thilikos.
\newblock {(Meta) Kernelization}.
\newblock {\em Journal of the ACM}, 63(5):44:1--44:69, 2016.

\bibitem[BFS10]{worker2010}
Hans~L. Bodlaender, Fedor~V. Fomin, and Saket Saurabh.
\newblock Open problems, {Worker} 2010, 2010.
\newblock
  \url{https://fpt.wdfiles.com/local--files/open-problems/open-problems.pdf}.

\bibitem[BCKP16]{BoralEtAl2016}
Anudhyan Boral, Marek Cygan, Tomasz Kociumaka, and Marcin Pilipczuk.
\newblock A fast branching algorithm for cluster vertex deletion.
\newblock {\em Theory of Computing Systems}, 58:357--376, 2016.

\bibitem[Cai96]{Cai1996}
Leizhen Cai.
\newblock Fixed-parameter tractability of graph modification problems for
  hereditary properties.
\newblock {\em Information Processing Letters}, 58(4):171--176, 1996.

\bibitem[{\v{C}}CS24]{CervenyChoudharySuchy2022}
Radovan {\v{C}}erven{\'y}, Pratibha Choudhary, and Ond{\v{r}}ej Such{\'y}.
\newblock On kernels for {$d$}-path vertex cover.
\newblock {\em Journal of Computer and System Sciences}, 144:103531, 2024.

\bibitem[CFK{\etalchar{+}}15]{CyganEtAl2015}
M.~Cygan, F.~V. Fomin, {\L}.~Kowalik, D.~Lokshtanov, D.~Marx, M.~Pilipczuk,
  M.~Pilipczuk, and S.~Saurabh.
\newblock {\em Parameterized Algorithms}.
\newblock Springer, 2015.

\bibitem[CFJ{\etalchar{+}}14]{fptSchool2014}
Marek Cygan, Fedor~V. Fomin, Bart M.~P. Jansen, {\L}ukasz Kowalik, Daniel
  Lokshtanov, D{\'a}niel Marx, Marcin Pilipczuk, Micha{\l} Pilipczuk, and Saket
  Saurabh.
\newblock Open problems for {FPT School} 2014, 2014.
\newblock \url{https://fptschool.mimuw.edu.pl/opl.pdf}.

\bibitem[DvM14]{DellVanMelkebeek2014}
H.~Dell and D.~van Melkebeek.
\newblock Satisfiability allows no nontrivial sparsification unless the
  polynomial-time hierarchy collapses.
\newblock {\em Journal of the ACM}, 61(4):23:1--23:27, 2014.

\bibitem[DM12]{DellMarx2012}
Holger Dell and D{\'a}niel Marx.
\newblock Kernelization of packing problems.
\newblock In {\em Proceedings of the Twenty-Third Annual ACM-SIAM Symposium on
  Discrete Algorithms ({SODA} 2012)}, pages 68--81. SIAM, 2012.

\bibitem[ER60]{ErdosRado1960}
Paul Erd{\H{o}}s and Richard Rado.
\newblock Intersection theorems for systems of sets.
\newblock {\em Journal of the London Mathematical Society}, 35(1):85--90, 1960.

\bibitem[FKN{\etalchar{+}}08]{fellows2008faster}
Michael~{R.} Fellows, Christian Knauer, Naomi Nishimura, Prabhakar Ragde, {F.}
  Rosamond, Ulrike Stege, Dimitrios~{M.} Thilikos, and Sue Whitesides.
\newblock Faster fixed-parameter tractable algorithms for matching and packing
  problems.
\newblock {\em Algorithmica}, 52:167--176, 2008.

\bibitem[FG06]{DBLP:series/txtcs/FlumG06}
J{\"{o}}rg Flum and Martin Grohe.
\newblock {\em Parameterized Complexity Theory}.
\newblock Texts in Theoretical Computer Science. An {EATCS} Series. Springer,
  2006.

\bibitem[FLL{\etalchar{+}}19]{FominEtAl2019Implicit}
F.~V. Fomin, T.-N. Le, D.~Lokshtanov, S.~Saurabh, S.~Thomass{\'e}, and
  M.~Zehavi.
\newblock Subquadratic kernels for implicit 3-hitting set and 3-set packing
  problems.
\newblock {\em ACM Transactions on Algorithms}, 15(1):13:1--13:44, 2019.

\bibitem[FLL{\etalchar{+}}23]{FominEtAl2023}
F.~V. Fomin, T.-N. Le, D.~Lokshtanov, S.~Saurabh, S.~Thomass{\'e}, and
  M.~Zehavi.
\newblock Lossy kernelization for (implicit) hitting set problems.
\newblock In {\em 31st Annual European Symposium on Algorithms (ESA 2023)},
  volume 274 of {\em LIPIcs}, pages 49:1--49:14, 2023.

\bibitem[FLSZ19]{FominEtAl2019Book}
F.~V. Fomin, D.~Lokshtanov, S.~Saurabh, and M.~Zehavi.
\newblock {\em Kernelization: Theory of Parameterized Preprocessing}.
\newblock Cambridge University Press, 2019.

\bibitem[GKK{\etalchar{+}}15]{GhoshEtAl2015Split}
Esha Ghosh, Sudeshna Kolay, Mrinal Kumar, Pranabendu Misra, Fahad Panolan,
  Ashutosh Rai, and M.~S. Ramanujan.
\newblock Faster parameterized algorithms for deletion to split graphs.
\newblock {\em Algorithmica}, 71(4):989--1006, 2015.

\bibitem[HKMN10]{HueffnerEtAl2010}
Falk H{\"u}ffner, Christian Komusiewicz, Hannes Moser, and Rolf Niedermeier.
\newblock Fixed-parameter algorithms for cluster vertex deletion.
\newblock {\em Theory of Computing Systems}, 47(1):196--217, 2010.

\bibitem[Jan26]{Jansen2026QuoVadis}
Bart M.~P. Jansen.
\newblock Quo vadis, kernelization?
\newblock {\em Computer Science Review}, 62:101035, 2026.

\bibitem[JR23]{JansenRoy2023}
Bart M.~P. Jansen and Shivesh~K. Roy.
\newblock Sunflowers meet sparsity: A linear-vertex kernel for weighted
  clique-packing on sparse graphs.
\newblock In {\em 18th International Symposium on Parameterized and Exact
  Computation (IPEC 2023)}, volume 285 of {\em LIPIcs}, pages 29:1--29:13,
  2023.

\bibitem[KL16]{KumarLokshtanov2016}
Mithilesh Kumar and Daniel Lokshtanov.
\newblock Faster exact and parameterized algorithm for feedback vertex set in
  tournaments.
\newblock In {\em 33rd Symposium on Theoretical Aspects of Computer Science
  (STACS 2016)}, volume~47 of {\em LIPIcs}, pages 49:1--49:13, 2016.

\bibitem[LX27]{LiuXiao2025}
Y.~Liu and M.~Xiao.
\newblock A refined kernel for {$d$-Hitting Set}.
\newblock {\em Journal of Computer and System Sciences}, 164:103861, 2027.

\bibitem[LPRS26]{lokshtanov2017lossy}
Daniel Lokshtanov, Fahad Panolan, M.~S. Ramanujan, and Saket Saurabh.
\newblock Lossy kernelization.
\newblock {\em SIAM Journal on Computing}, 55(5):1032--1097, 2026.

\bibitem[Mos10]{Moser2010}
Hannes Moser.
\newblock {\em Finding Optimal Solutions for Covering and Matching Problems}.
\newblock Cuvillier Verlag, G{\"o}ttingen, 2010.

\bibitem[PRST19]{PhilipEtAl2019}
Geevarghese Philip, Varun Rajan, Saket Saurabh, and Prafullkumar Tale.
\newblock Subset feedback vertex set in chordal and split graphs.
\newblock {\em Algorithmica}, 81:3586--3629, 2019.

\bibitem[Sin05]{Sinz2005}
C.~Sinz.
\newblock Towards an optimal {CNF} encoding of boolean cardinality constraints.
\newblock In {\em Principles and Practice of Constraint Programming ({CP}
  2005)}, volume 3709 of {\em LNCS}, pages 827--831, 2005.

\bibitem[vB14]{vanBevern2014}
R.~van Bevern.
\newblock Towards optimal and expressive kernelization for {$d$-Hitting Set}.
\newblock {\em Algorithmica}, 70:129--147, 2014.

\end{thebibliography}

\end{document}